\documentclass[a4paper,11pt]{article}

\usepackage{amsfonts,amssymb,amsmath,amsthm,dsfont,xfrac,xspace}
\usepackage{fullpage}
\usepackage{graphicx}
\usepackage{subfig}
\usepackage{cite}
\usepackage{hyperref}
\usepackage{algorithmic}
\graphicspath{{./},{fig/}}
\makeatletter
\let\NAT@parse\undefined
\makeatother
\usepackage[sort&compress, numbers]{natbib}
\usepackage{hyperref}

\newsavebox{\ieeealgbox}
\newenvironment{boxedalgorithmic}
  {\begin{lrbox}{\ieeealgbox}
   \begin{minipage}{\dimexpr\columnwidth-2\fboxsep-2\fboxrule}
   \begin{algorithmic}[1]}
  {\end{algorithmic}
   \end{minipage}
   \end{lrbox}\noindent\fbox{\usebox{\ieeealgbox}}}

\newcommand{\bsy}[1]{\boldsymbol{#1}}
\newcommand{\supp}{\text{supp}}
\newcommand\independent{\perp\!\!\!\perp}

\DeclareMathOperator*{\argmin}{arg\,min}

\newtheorem{thm}{Theorem}
\newtheorem{lem}{Lemma}[thm]
\renewcommand{\thelem}{\the \numexpr (\value{thm}+1) \relax.\arabic{lem}}
\newtheorem{corr}{Corollary}[thm]

\newtheorem{lemSA}{Lemma}

\newcounter{algoCounter}

\title{Simultaneous Orthogonal Matching Pursuit With Noise Stabilization: Theoretical Analysis} 

\author{ Jean-Fran\c{c}ois Determe\thanks{Jean-Fran\c{c}ois Determe and Fran\c{c}ois Horlin are with the OPERA Wireless Communications Group, Universit\'e Libre de Bruxelles, 1050 Brussels, Belgium. E-mail: jdeterme@ulb.ac.be, fhorlin@ulb.ac.be. Jean-Fran\c{c}ois Determe is funded by the Belgian National Science Foundation (F.R.S.-FNRS).}  \quad J\'er\^{o}me Louveaux\footnotemark[2] \quad  Laurent Jacques\thanks{Laurent Jacques and J\'{e}r\^{o}me Louveaux are with the ICTEAM departement, Universit\'e Catholique de Louvain. Laurent Jacques is funded by the Belgian National Science Foundation (F.R.S.-FNRS).} \quad Fran\c{c}ois Horlin\footnotemark[1] }

\begin{document}
\maketitle

\begin{abstract}
This paper studies the joint support recovery of similar sparse vectors on the basis of a limited number of noisy linear measurements, \textit{i.e.}, in a multiple measurement vector (MMV) model. The additive noise signals on each measurement vector are assumed to be Gaussian and to exhibit different variances. The simultaneous orthogonal matching pursuit (SOMP) algorithm is generalized to weight the impact of each measurement vector on the choice of the atoms to be picked according to their noise levels. The new algorithm is referred to as SOMP-NS where NS stands for noise stabilization.

To begin with, a theoretical framework to analyze the performance of the proposed algorithm is developed. This framework is then used to build conservative lower bounds on the probability of partial or full joint support recovery. Numerical simulations show that the proposed algorithm outperforms SOMP and that the theoretical lower bound provides a great insight into how SOMP-NS behaves when the weighting strategy is modified.
\end{abstract}

\section{Introduction}
The recovery of sparse signals of high dimensions on the basis of noisy linear measurements is an important problem in the field of signal acquisition and processing. When the number of linear observations is significantly lower than the dimension of the signal to be recovered, the signal recovery may exploit the property of sparsity to deliver correct results. The field of research that studies such problems is often referred to as \textit{compressed sensing} or \textit{compressive sensing} (CS) \cite{donoho2006compressed}.\\

Several computationally tractable methods to address CS problems have been developed in the last two decades \cite{davis1997adaptive, cai2011orthogonal, cotter2005sparse, leviatan2006simultaneous, needell2010signal, tropp2005simultaneous, tropp2006algorithms, tropp2006just, tropp2007signal}. Among them, greedy methods prove to be valuable choices as their complexity is significantly lower than that of algorithms based on $\ell_1$-minimization \cite{tropp2007signal}. \\

While many CS problems involve only one sparse signal and the corresponding \textit{measurement vector}, \textit{i.e.}, the vector gathering all the linear observations of this signal, some applications either require or at least benefit from the presence of several sparse signals and measurement vectors. Examples of such applications are available in Section~\ref{subsec:applications}. Models involving one measurement vector are referred to as single measurement vector (SMV) models while multiple measurement vector (MMV) models involve at least two measurement vectors \cite{eldar2009robust}.\\

When the supports of the sparse signals are similar, it is possible to improve the reliability of the recovery by making joint decisions  to determine the estimated support \cite{chen2006theoretical, gribonval2008atoms}. Thereby, all the measurement vectors intervene in the estimation of the support and the final support is common to all the sparse vectors. Algorithms performing joint recovery are also capable to weaken the influence of additive measurement noise on the performance provided that the noise signals are statistically independent and exhibit some degree of isotropy.\\

Orthogonal matching pursuit (OMP) is one of the most extensively used greedy algorithm designed to solve SMV problems \cite{mallat1993matching}.  Among several greedy algorithms conceived to deal with multiple measurement vectors, the extension of OMP to the MMV paradigm, referred to as simultaneous orthogonal matching pursuit (SOMP), is of great interest as it remains simple, both conceptually and algorithmically \cite{tropp2006algorithms}.

\subsection{Motivation \& Objective}

The classical SOMP algorithm does not account for the possibly different measurement vector noise levels. In some sense, all the measurement vectors are considered equally worthy. However, it is clear that an optimal joint support recovery method should necessarily take into account the noise levels by accordingly weighting the impact of each measurement vector on the decisions that are taken. The first aim of this paper is to extend SOMP by gifting it with weighting capabilities. The new algorithm will be referred to as SOMP with noise stabilization (SOMP-NS) and basically extends the decision metric of SOMP to weight the impact of each measurement vector onto the decisions that are taken.\\

The second objective is to provide theoretical and numerical evidence that the proposed algorithm indeed enables one to achieve higher performance than the other greedy alternatives when the noise levels, or more generally the signal-to-noise ratios (SNR), vary from one measurement vector to another. 

\subsection{Detailed contribution}

We study partial and full support recovery guarantees of SOMP-NS for a MMV signal model incorporating arbitrary sparse signals to be recovered and statistically independent  additive Gaussian noise vectors exhibiting diagonal covariance matrices, \textit{i.e.}, the entries within each vector are statistically independent. It is assumed that the variances of the entries within each noise vector are identical although they may be different for each measurement vector. The signal model is thoroughly detailed in Section~\ref{subsec:signalModel}. \\

Our first contribution is the proposal of SOMP-NS which generalizes SOMP by weighting the measurement vectors. The second contribution is a novel theoretical analysis of SOMP and SOMP-NS in the presence of additive Gaussian noise on the measurements. To the best of the authors' knowledge, the theoretical analysis in this paper has never been proposed, neither for SOMP nor for SOMP-NS.\\

Finally, numerical simulations will show that the weighting capabilities of SOMP-NS enable one to improve the performance with regards to SOMP when the noise vectors exhibit different powers. The numerical results will also provide evidence that the theoretical analysis accurately depicts key characteristics of SOMP-NS. In particular, closed-form formulas for the optimal weights will be derived from the theoretical analysis and will be compared to the simulation results.

\subsection{Related work}

Several authors have worked on similar problems. The study of full support recovery guarantees for OMP with $\ell_2$ or $\ell_{\infty}$-bounded noises as well as with Gaussian noises has been performed in \cite{cai2011orthogonal}. The authors of \cite{cai2011orthogonal} also provided conditions on the stopping criterion to ensure that OMP stops after having picked all the correct atoms. \\

Our analysis is similar to that performed by Tropp in \cite{tropp2006just} for convex programming methods in a SMV setting. Together with Gilbert \cite{tropp2007signal}, they analyzed the probability of full support recovery by means of OMP for Gaussian measurement matrices in the noiseless case. Their result has subsequently been refined by Fletcher and Rangan in \cite{fletcher2012orthogonal, rangan2009orthogonal} to account for additive measurement noise by means of a high-SNR analysis, \textit{i.e.}, it is assumed that the signal-to-noise ratio scales to infinity. All of the papers discussed so far only focus on the SMV framework.\\

The theoretical analysis of our paper is partially inspired from \cite{cai2011orthogonal} and has been generalized to the MMV framework. It is worth pointing out that our analysis does not require the high SNR assumption of \cite{fletcher2012orthogonal, rangan2009orthogonal}, properly captures the benefits provided by multiple measurement vectors but nevertheless presents some inaccuracies that are to be discussed at the end of this paper.\\

Gribonval \textit{et al.} have performed an analysis of SOMP  for a problem similar to ours in \cite{gribonval2008atoms}. They were interested in providing a lower bound on the probability of correct support recovery when the signal to be estimated is sparse and its non-zero entries are statistically independent mean-zero Gaussian random variables exhibiting possibly different variances.\\

While our statistical analysis considers the additive measurement noise as a random variable and the sparse signals to be recovered as deterministic quantities, the results obtained in \cite{gribonval2008atoms} use the opposite approach, \textit{i.e.}, the sparse signals are random and the noise is deterministic. Thus, the problem addressed in our paper differs from that presented in \cite{gribonval2008atoms} but both papers use similar mathematical tools and the criteria to ensure full support recovery with high probability share analogous properties. This last remark will be further discussed in Section~\ref{subsec:GribonvalrelatedThm}.

\subsection{Outline}

First of all, Section~\ref{sec:sigModel} progressively introduces the context, provides a detailed description of the signal model and depicts an associated application. Afterwards, Section~\ref{sec:SOMPandSOMPNS} provides descriptions of SOMP and SOMP-NS.\\

Before deriving the theoretical analysis, Section~\ref{sec:background} introduces the mathematical tools necessary for its execution. Section~\ref{sec:section3} then provides general theoretical results on the proper recovery of sparse vectors by means of SOMP-NS.
On the basis of the results from Section~\ref{sec:section3}, we show in Section~\ref{sec:recovGuarantees} that, for Gaussian noises,  the probability of failure of SOMP-NS decreases exponentially with regards to the number of measurement vectors. \\

In Section~\ref{sec:numresults}, extensive numerical simulations show that adequate weighting strategies enable SOMP-NS to outperform SOMP whenever the noise variances for each measurement vector are different. Also, a closed-form weighting strategy is derived from the theoretical analysis of the previous sections and these weights are compared to the optimal ones obtained by simulation. Finally, the simulation results show which aspects of the behavior of SOMP-NS are properly captured by the proposed theoretical analysis. The reasons why our analysis fails to capture some characteristics of SOMP are discussed and potential workarounds are proposed for investigation. 

\subsection{Conventions}

We find preferable to introduce here the common notations used in this paper. First of all, $ \forall n \in \mathbb{N}, \left[n\right] := \left\lbrace 1, 2, \dots, n \right\rbrace$. For any set $A$, $|A|$ refers to its cardinality. Matrices are denoted by upper case bold letters while vectors are written in lower case bold letters. The tranpose operation is denoted by the superscript ${ }^{\mathrm{T}}$. For $\bsy{x} \in \mathbb{R}^n$, $x_i$ denotes the $i$-th component of $\bsy{x}$. Similarly, the $j$-th column vector of any matrix is denoted by the corresponding lower case bold letter with subscript $j$, \textit{e.g.}, the $j$-th column of $\bsy{\Phi}$ is $\bsy{\phi}_j$. The $(i,j)$-th entry of a matrix is denoted by the upper case letter with subscript $i,j$, \textit{e.g.}, the $(i,j)$-th entry of $\bsy{\Phi}$ is written $\Phi_{i,j}$. The $\ell_p$ norm ($1 \leq p < \infty$) of $\bsy{\phi} \in \mathbb{R}^m$ is defined as $\| \bsy{\phi} \|_p := \left( \sum_{i=1}^{m} | \phi_i |^p \right)^{1/p}$. Moreover, the $\ell_{\infty}$ norm of vector $\bsy{\phi} \in \mathbb{R}^m$ is defined as $\|\bsy{\phi}\|_{\infty} := \max_{i \in \left[m \right]} |\phi_i|$. For any matrix $\bsy{\Phi} \in \mathbb{R}^{m \times n}$, $\mathrm{span}(\bsy{\Phi})$ denotes the space spanned by its column vectors. Also, the trace of $\bsy{\Phi}$ is $\mathrm{tr}(\bsy{\Phi})$ and the Frobenius norm of $\bsy{\Phi}$ is denoted by $ \| \bsy{\Phi} \|_{\mathrm{F}}$. For $S \subseteq \lbrack n \rbrack$, $\bsy{\Phi}_{S}$ denotes the submatrix of $\bsy{\Phi}$ that comprises its columns indexed by $S$. Likewise, $\bsy{x}_{S}$ is the subvector of $\bsy{x}$ comprising only the components indexed within $S$. The Moore-Penrose pseudoinverse of any matrix $\bsy{\Phi}$ is denoted by $\bsy{\Phi}^+$. The orthogonal complement of a vector subspace $A$ is given by $A^{\perp}$. For any random variable $X$, its cumulative density function (CDF) is denoted by $F_X$ while its probability density function (PDF) is written $f_X$ (when it exists). Similarly, the joint CDF and PDF of the random variables $X_1, \dots, X_K$ are written as $F_{X_1, \dots, X_K}$ and $f_{X_1, \dots, X_K}$ respectively. The probability measure is given by $\mathbb{P}$ while the mathematical expectation is denoted by $\mathbb{E}$. The statistical independence symbol is written $\independent$.

\section{Signal model}\label{sec:sigModel}

\subsection{Context}

We define the support of a vector $\bsy{x}$ as $\mathrm{supp}(\bsy{x}) := \lbrace i : x_i \neq 0 \rbrace$. The $\ell_0$ ``norm''\footnote{This is an abuse of language as this mathematical object does not satisfy the properties of a norm} of a vector $\bsy{x}$ is defined as $\| \bsy{x} \|_0 = | \mathrm{supp}(\bsy{x}) |$.  Loosely speaking, we say that $\bsy{x} \in \mathbb{R}^n$ is sparse whenever $\| \bsy{x} \|_0 \ll n$. Moreover, $\bsy{x}$ is said to be $s$-sparse whenever $\| \bsy{x} \|_0 \leq s$. \\

Let us consider a collection of $K$ signals $\bsy{f}_k \in \mathbb{R}^{n}$ ($1 \leq k \leq K$) that are sparse in an orthonormal basis, \textit{i.e.}, $\bsy{f}_k = \bsy{\Psi} \bsy{x}_k$ where  $\bsy{\Psi} \in \mathbb{R}^{n \times n}$ represents the orthonormal basis and $\bsy{x}_k \in \mathbb{R}^{n}$ is the sparse coefficient vector of $\bsy{f}_k$ expressed in $\bsy{\Psi}$.\\

We now consider a unique linear measurement process to recover each one of the $K$ sparse coefficient vector $\bsy{x}_k$. Moreover, the measurement process is assumed to deliver a number of observations $m$ significantly lower than $n$. Additive measurement noise vectors $\bsy{e}_k \in \mathbb{R}^{m}$ are also accounted for. Formally, the latter statements rewrite
\begin{equation}\label{eq:introSigModel1}
\bsy{y}_k = \bsy{\Phi} \bsy{\Psi} \bsy{x}_k + \bsy{e}_k
\end{equation}
where the \textit{measurement matrix} $\bsy{\Phi} \in \mathbb{R}^{m \times n}$ denotes the linear measurement process and the \textit{measurement vectors} $\lbrace \bsy{y}_k \rbrace_{1 \leq k \leq K}$ gather the observations.\\ 

Since the number of observations $m$ is lower than $n$, arbitrary vectors $\bsy{x}_k$ cannot be recovered from $\bsy{y}_k$, even in the noiseless case. However, in the noiseless case, it has been shown \cite{candes2005decoding, rudelson2005geometric, donoho2006compressed, candes2006near, candes2006robust} that $\bsy{x}_k$ can be recovered provided that it is sparse enough, \textit{i.e.}, the cardinality of its support is below a certain threshold, and that the measurement matrix $\bsy{\Phi}$ exhibits specific properties such as the restricted isometry property that is described afterwards.\\

The orthonormal basis $\bsy{\Psi}$ is often assumed to be the canonical basis. The reason that explains this simplification relies on the fact the the measurement matrix $\bsy{\Phi}$ is usually generated as a realization of a subgaussian random matrix. It can be shown that such random matrices are well-condtionned for sparse support recovery, even when multiplied by orthonormal matrices, \textit{i.e.}, $\bsy{\Phi} \bsy{\Psi} = \bsy{\Phi}'$ remains a subgaussian random matrix. This phenomenon is more thoroughly discussed in \cite[Section 9.1]{foucart2013mathematical}. Thereby, the signal model will be simplified to
\begin{equation}\label{eq:introSigModel2}
\bsy{y}_k = \bsy{\Phi} \bsy{x}_k + \bsy{e}_k.
\end{equation}


If several sparse signals $\bsy{x}_k$ share similar supports, then it is interesting to simultaneously recover their joint support $\bigcup_{1 \leq k \leq K} \mathrm{supp}(\bsy{x}_k)$ instead of performing $K$ independent and possibly different estimations of the support of every vector $\bsy{x}_k$. The reason that explains why such a strategy yields performance improvements is twofold:
\begin{enumerate}
\item It has been shown \cite{gribonval2008atoms} that when the sparse vectors $\bsy{x}_k$ share a similar support whose associated entries are highly variable from one vector to another, then the probability of correct support recovery increases with the number of available measurement vectors $K$.
\item In the noisy setting, the vectors $\bsy{e}_k$ are often statistically independent and isotropic which suggests that a joint support recovery procedure could be capable of filtering them. This property is  part of what the theoretical results and simulations of this paper establish.
\end{enumerate}

Once the joint support has been recovered, $\bsy{x}_k$ is easily recovered by solving a least squares problem involving only the non-zero entries of $\bsy{x}_k$ and the associated column vectors from $\bsy{\Phi}$.

\subsection{Signal model} \label{subsec:signalModel}

We wish to provide here a formal and precise statement of the signal model to be used in the rest of this paper.\\

As mentioned previously, we consider $K$ measurement vectors $\bsy{y}_k$ that are generated on the basis of $K$ sparse signals $\bsy{x}_k$ whose supports are similar. 
We consider the following MMV model:

\begin{equation}\label{eq:MMVSignalModel}
\bsy{Y} = \bsy{\Phi} \bsy{X} + \bsy{E}
\end{equation}
where $\bsy{Y} = \big( \bsy{y}_1, \cdots, \bsy{y}_K \big)$ is composed of the $K$ measurement vectors $\bsy{y}_k \in \mathbb{R}^m$. $\bsy{X} = \big(\bsy{x}_1, \cdots, \bsy{x}_K \big)$ comprises the $K$ sparse signals $\bsy{x}_k \in \mathbb{R}^n$. Finally, the columns of $\bsy{E} = \big( \bsy{e}_1, \cdots, \bsy{e}_K \big)$ correspond to measurement errors. Each error vector $\bsy{e}_k$ is distributed as $\mathcal{N}(0, \sigma_k^2 \bsy{I}_{m \times m})$ and all of them are statistically independent. The purpose of (\ref{eq:MMVSignalModel}) is to aggregate the $K$ equations defined in (\ref{eq:introSigModel2}) into a single relation. This representation will be preferred throughout the rest of this paper. \\

Before going further, we will point out that the mathematical problem of the joint support recovery is equivalent to finding the columns of $\bsy{\Phi}$ that enable one to fully express $\bsy{\Phi} \bsy{X}$. Thereby, $\bsy{\Phi}$ may be seen as a \textit{dictionary matrix} whose columns $\lbrace \bsy{\phi}_j \rbrace_{j \in \lbrack n \rbrack}$ are the \textit{atoms} of the associated dictionary. The problem of joint support recovery then boils down to determining which atoms to choose to simultaneously express the $K$ measurement vectors as their linear combinations. Although viewing $\bsy{\Phi}$ as a measurement process is well suited to the description of typical applications, the dictionary approach is more appropriate for the sake of presenting the mathematical results and will thus be adopted for the rest of this paper. However, the term \textit{measurement vector} will not be replaced to stay consistent with the standard MMV terminology. \\

The dictionary matrix $\bsy{\Phi}$ is assumed to satisfy the restricted isometry property of an order equal to or higher than the cardinality of the joint support $S = \cup_{k=1}^{K} \mathrm{supp}(\bsy{x}_k)$. Moreover, it will be assumed that each column of this matrix exhibits a unit $\ell_2$ norm. We briefly review two procedures to obtain a dictionary matrix that satisfies both properties above with high probability:
\begin{enumerate}
\item Generate a random Gaussian matrix and then normalize the $\ell_2$ norm of its columns. In such a way, the atoms are uniformly distributed on the unit hypersphere $\mathbb{S}^{n-1}$ of dimension $n$. It is possible to show that this class of random matrices satisfies the restricted isometry property with high probability (see \cite{adamczak2011restricted} for example).
\item Generate a matrix whose entries are statistically independent Rademacher random variables. Each column of the resulting matrix is then normalized by multiplying the matrix by $1/\sqrt{m}$ to obtain atoms exhibiting unit $\ell_2$ norms.
\end{enumerate}

\subsection{Applications} \label{subsec:applications}

The typical scenario associated with signal models (\ref{eq:introSigModel2}) and (\ref{eq:MMVSignalModel}) is depicted in Figure~\ref{fig:FirstScenario}. The idea is to observe a physical quantity, \textit{e.g.}, a chemical composition, a wireless signal, etc.  at different locations and/or time instants by means of $K$ \textit{sensing nodes} $N_k$ whose only purpose is to acquire observations, \textit{i.e.}, measurement vectors, and repatriate them to a central node (CN) that will simultaneously process all the data. In such a configuration, the sensing nodes are generally cheap and exhibit very limited computational capabilities and power consumption while the central node is more costly because it has to deliver higher performance.\\


\begin{figure}[!h]
\centering
\includegraphics[scale=1.1]{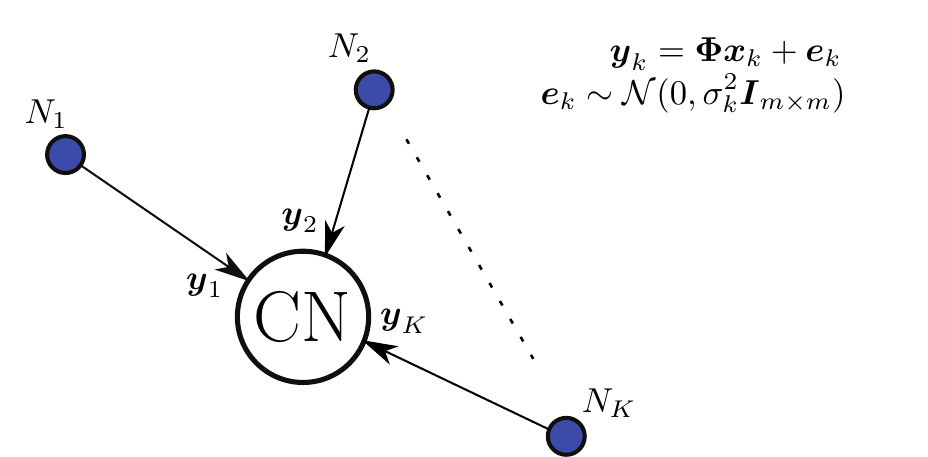}
\caption{First scenario -- $K$ nodes with different noise levels generate $K$ measurement vectors on which joint estimation of the sparse support $\cup_{k=1}^{K} \mathrm{supp}(\bsy{x}_k)$ can be performed}
\label{fig:FirstScenario}
\end{figure}

Although several applications of the signal model (\ref{eq:introSigModel2}) are presented in \cite[Section 3.3]{baron2009distributed}, we propose to take as an example the spectrum sensing problem. Spectrum sensing aims at scanning multi-gigahertz electromagnetic (EM) spectrums at a rate that is below that of Nyquist. The reason that motivates this objective is that, although most of the available frequency bands are licensed to specific users and thus costly to acquire, it has been observed that the spectrum occupancy is limited, \textit{i.e.}, the spectrum is sparse in the frequency domain at a given time instant and location. Therefore, it would be interesting to observe this spectrum through an appropriate linear measurement process, as described in Equation~(\ref{eq:introSigModel2}), and then use algorithms tailored for CS problems to determine, in real-time, which frequency bands are free and can be used to transmit information.\\

In this application, each entry of the sparse signals $\bsy{x}_k$ represents the power of a given frequency band. Since most of the spectrum is assumed to be unused at a given time instant and location, the vectors $\bsy{x}_k$ are expected to be sparse. Although the nodes should ideally be exposed to the same spectrum, this is not the case in practice because of the Rayleigh fading that strongly attenuates some frequency bands, thus being invisible to some nodes. In multiple input multiple output (MIMO) wireless communications, this issue is circumvented by placing the receivers sufficiently far away from one another so that the fading becomes statistically independent for each node and thus highly unlikely to strongly attenuate the same frequency band for every node. In a likewise fashion, the same solution will work for our framework since occasional ``holes'' will reveal to be a nonissue when performing joint decisions.\\

Finally, the different nodes may exhibit different noise levels because of discrepancies in the fabrication process or because the hardware components (\textit{e.g.}, amplifiers, multipliers, filters) of each node are different. This last remark justifies the multiple noise variances hypothesis of this paper.

\section{Towards a weighted greedy algorithm}\label{sec:SOMPandSOMPNS}

The two next subsections present two methods for addressing the problem of joint support recovery envisioned in Section~\ref{subsec:signalModel}. The first method, SOMP, is standard in the literature but does not include noise stabilization. The second method, our contribution, generalizes the first one by multiplying each measurement vector $\bsy{y}_k$ ($1 \leq k \leq K$) by a weight $q_k \geq 0$.

\subsection{Simultaneous orthogonal matching pursuit}

The original OMP algorithm \cite{mallat1993matching, tropp2004greed} has been generalized in several ways to deal with matrix signals $\bsy{Y} \in \mathbb{R}^{m \times K}$,  \textit{i.e.}, MMV problems. Simultaneous orthogonal matching pursuit (SOMP) is a possible generalization of OMP \cite{chen2006theoretical, tropp2005simultaneous, tropp2006algorithms}.\\

SOMP is a greedy algorithm that provides approximate solutions to the joint support recovery problem by successively picking atoms from $\bsy{\Phi}$ to simultaneously approximate the $K$ measurement vectors $\bsy{y}_k \in \mathbb{R}^{n}$. SOMP is described in Algorithm \ref{alg:SOMP}. \\

\begin{figure}[!h]
\textsc{Algorithm \refstepcounter{algoCounter}\label{alg:SOMP}\arabic{algoCounter}}:\\ 
Simulatenous orthogonal matching pursuit (SOMP)\\

\vspace{-2mm}
\begin{boxedalgorithmic}
\small
\REQUIRE $\bsy{Y} \in \mathbb{R}^{m \times K}$, $\bsy{\Phi} \in \mathbb{R}^{m \times n}$, $s \geq 1$
\STATE Initialization: $\bsy{R}^{(0)} \leftarrow \bsy{Y}$ and $S_0 \leftarrow \emptyset$
\STATE $t \leftarrow 0$
\WHILE{$t < s$}
\STATE Determine the atom of $\bsy{\Phi}$ to be included in the support: \\ $j_t \leftarrow \mathrm{argmax}_{j} ( \| (\bsy{R}^{(t)})^{\mathrm{T}} \bsy{\phi}_j \|_1 )$
\STATE Update the support : $S_{t+1} \leftarrow S_{t} \cup \left\lbrace j_t \right\rbrace$
\STATE Projection of each measurement vector onto $\mathrm{span}(\boldsymbol{\Phi}_{S_{t+1}})$: \\$\bsy{Y}^{(t+1)} \leftarrow \boldsymbol{\Phi}_{S_{t+1}} \boldsymbol{\Phi}_{S_{t+1}}^{+} \bsy{Y}$
\STATE Projection of each measurement vector onto $\mathrm{span}(\boldsymbol{\Phi}_{S_{t+1}})^{\perp}$~: \\ $\bsy{R}^{(t+1)} \leftarrow \bsy{Y} - \bsy{Y}^{(t+1)}$
\STATE $t \leftarrow t + 1$
\ENDWHILE
\RETURN $S_s$ \COMMENT{Support at last step}
\end{boxedalgorithmic}
\end{figure}

We now explain how SOMP proceeds. The residual at iteration $t$, denoted by $\bsy{R}^{(t)}$, consists of the projection of each one of the original signals $\bsy{y}_k$ onto the orthogonal complement of $\mathrm{span}(\bsy{\Phi}_{S_t})$. In such a way, the residual is orthogonal to every atom that has been chosen so far. Initially, the residual is chosen equal to the original signal. The decision on which atom to choose (step 4) is based on the sum of the inner products of every atom $\bsy{\phi}_{j}$  with each residual measurement vector $\bsy{r}^{(t)}_k$ (where $\bsy{r}^{(t)}_k$ refers to the $k$-th column of $\bsy{R}^{(t)}$) since
\begin{equation}
\| (\bsy{R}^{(t)})^{\mathrm{T}} \bsy{\phi}_j \|_1 = \sum_{k=1}^{K} | \langle \bsy{\phi}_{j},  \bsy{r}^{(t)}_k \rangle |.
\end{equation}

The index of the atom maximizing the $\ell_1$ norm is included in the support (step 5). Then, the original signal $\bsy{Y}$ is projected onto the  orthogonal complement of $\mathrm{span}(\boldsymbol{\Phi}_{S_{t+1}})$ (steps 6 and 7). \\

In this setting, SOMP stops after exactly $s$ iterations. However, it is worth mentioning that the stopping criterion usually comprises a criterion based on the number of iterations as well as another one relying on the norm of the residual, \textit{i.e.}, if the norm of the residual is below a certain threshold, then OMP stops. Different norms can be used for the second criterion but these considerations will not be further discussed in this paper. The interested reader can consult \cite{cai2011orthogonal} for related matters.  \\

Furthermore, maximizing the $\ell_1$ norm in step $4$ is not the unique choice. Other authors have investigated different norms, \textit{e.g.}, the $\ell_2$ and $\ell_{\infty}$ norms. Nevertheless, some numerical simulations reveal that the choice of the norm has very little effect on the performance (see \cite[Figure 3]{chen2006theoretical}).

\subsection{SOMP-NS: A weighted greedy algorithm}

We now present the development of a noise stabilization strategy to be used in conjunction with SOMP that has low computational requirements. The equivalent new algorithm is referred to as SOMP-NS where NS stands for \textit{noise stabilization}. Algorithm~\ref{alg:SOMPNS} describes the first form of SOMP-NS.\\

This novel algorithm is a  generalization of SOMP that weights the impact of each measurement vector within matrix $\bsy{Y} \in \mathbb{R}^{m \times K}$ on the decisions performed at each iteration.\\

\begin{figure}[!h]
\textsc{Algorithm \refstepcounter{algoCounter}\label{alg:SOMPNS}\arabic{algoCounter}}:\\ 
SOMP with noise stabilization (SOMP-NS) -- First form\\

\vspace{-2mm}
\begin{boxedalgorithmic}
\small
\REQUIRE $\bsy{Y} \in \mathbb{R}^{m \times K}$, $\bsy{\Phi} \in \mathbb{R}^{m \times n}$, $\lbrace q_k \rbrace_{1 \leq k \leq K}$, $s \geq 1$
\STATE $\bsy{R}^{(0)} \leftarrow \bsy{Y}$ and $S_0 \leftarrow \emptyset$ \COMMENT{Initialization}
\STATE $t \leftarrow 0$
\WHILE{$t < s$}
\STATE Determine the atom of $\bsy{\Phi}$ to be included in the support:\\
$j_t \leftarrow \mathrm{argmax}_{j} \left( \sum_{k=1}^{K} \left| \left\langle \bsy{r}^{(t)}_k , \bsy{\phi}_{j} \right\rangle  \right| q_k \right)$
\STATE $S_{t+1} \leftarrow S_{t} \cup \left\lbrace j_t \right\rbrace$
\STATE Projection of each measurement vector onto $\mathrm{span}(\boldsymbol{\Phi}_{S_{t+1}})$:\\
$\bsy{Y}^{(t+1)} \leftarrow \boldsymbol{\Phi}_{S_{t+1}} \boldsymbol{\Phi}_{S_{t+1}}^{+} \bsy{Y}$
\STATE Projection of each measurement vector onto $\mathrm{span}(\boldsymbol{\Phi}_{S_{t+1}})^{\perp}$:\\
$\bsy{R}^{(t+1)} \leftarrow \bsy{Y} - \bsy{Y}^{(t+1)}$
\STATE $t \leftarrow t + 1$
\ENDWHILE
\RETURN $S_s$ \COMMENT{Support at last step}
\end{boxedalgorithmic}
\end{figure}

SOMP-NS is actually very close to SOMP. Both algorithms decide on which atom to pick on the basis of a sum of absolute values of inner products, each term in the sum depending only on one measurement vector. SOMP gives the same importance to each measurement vector whereas its weighted counterpart introduces weights $q_k \geq 0$ ($1 \leq k \leq K$) so as to give more or less importance to each measurement vector.\\

A second form of SOMP-NS that is more computationally efficient is available in Algorithm \ref{alg:SOMPNSv2}. In the second form, \textbf{SOMP}($\bsy{Y}, \bsy{\Phi}, s$) refers to the regular SOMP algorithm described in Algorithm \ref{alg:SOMP}.

\begin{figure}[!h]
\textsc{Algorithm \refstepcounter{algoCounter}\label{alg:SOMPNSv2}\arabic{algoCounter}}:\\ 
SOMP with noise stabilization (SOMP-NS) -- Second form\\

\vspace{-2mm}
\begin{boxedalgorithmic}
\small
\REQUIRE $\bsy{Y} \in \mathbb{R}^{m \times K}$, $\bsy{\Phi} \in \mathbb{R}^{m \times n}$, $\lbrace q_k \rbrace_{1 \leq k \leq K}$, $s \geq 1$
\STATE $\bsy{Q} \leftarrow \mathrm{diag}(q_1, \dots, q_K)$
\STATE The columns of $\bsy{Y}$ are weighted beforehand: $\bsy{Y} \leftarrow \bsy{Y} \bsy{Q}$ 
\STATE Apply the regular SOMP algorithm: $S \leftarrow $ \textbf{SOMP}($\bsy{Y}, \bsy{\Phi}, s$)
\RETURN $S$
\end{boxedalgorithmic}
\end{figure}

Both forms are equivalent since $ | \langle \bsy{r}^{(t)}_k , \bsy{\phi}_{j} \rangle  | q_k = | \langle \bsy{r}^{(t)}_k q_k , \bsy{\phi}_{j} \rangle |   = | \langle (\bsy{R}^{(t)} \bsy{Q})_k , \bsy{\phi}_{j} \rangle | $ and $(\bsy{R}^{(t)} \bsy{Q})_k = ( (\bsy{I} - \boldsymbol{\Phi}_{S_{t}} \boldsymbol{\Phi}_{S_{t}}^{+})\bsy{Y} \bsy{Q})_k = (\bsy{I} - \boldsymbol{\Phi}_{S_{t}} \boldsymbol{\Phi}_{S_{t}}^{+}) (\bsy{Y} \bsy{Q})_k$. The last equality holds true because the orthogonal projector $(\bsy{I} - \boldsymbol{\Phi}_{S_{t}} \boldsymbol{\Phi}_{S_{t}}^{+})$ is applied to each column of $\bsy{Y} \bsy{Q}$ separately. Although the residual matrix $\bsy{R}^{(t)}$ is different for the two forms of SOMP-NS, the difference only consists in a multiplicative term $q_k$ for each column $k$ of $\bsy{R}^{(t)}$ and it does not modify the atoms added to the estimated support.

\section{Technical background for the theoretical analysis}\label{sec:background}

This section briefly explains the mathematical tools needed to conduct the theoretical analysis.
\subsection{Dictionary and coherence}

Let $\bsy{\Phi} \in \mathbb{R}^{m \times n}$ be a matrix composed of $n$ column vectors $\bsy{\phi}_j \in \mathbb{R}^{m}$. Moreover, $\forall j \in \left[n\right]$, $\|\bsy{\phi}_j\|_2 = 1$.  The coherence $\mu(\bsy{\Phi})$ of $\bsy{\Phi}$ is defined as
 
\begin{equation}
\mu(\bsy{\Phi}) := \max_{i \neq j} \left| \left\langle \bsy{\phi}_i, \bsy{\phi}_j \right\rangle \right| = \max_{i \neq j} \left| \bsy{\phi}_i^{\mathrm{T}} \bsy{\phi}_j \right|
\end{equation}
Similarly, one can define the cumulative coherence function (also referred to as the Babel function) of $\bsy{\Phi}$ as
\begin{equation}
\mu_1(\bsy{\Phi}, p) = \max_{\Lambda, \left| \Lambda \right| = p} \max_{j \not\in  \Lambda} \sum_{i \in \Lambda} \left| \langle \bsy{\phi}_i, \bsy{\phi}_j \rangle \right|.
\end{equation}
Trivially, $\mu(\bsy{\Phi}) = \mu_1(\bsy{\Phi}, 1)$. For the sake of brevity, $\mu(\bsy{\Phi})$ and $\mu_1(\bsy{\Phi}, p)$ will be respectively denoted as $\mu$ and $\mu_1(p)$ from now on.\\
It is also frequent to use the bound
\begin{equation}\label{eq:mu1tomuBound}
\mu_1(p) \leq p \mu.
\end{equation}

\subsection{Support and norms}

If $\bsy{U} = \big( \bsy{u}_1, \; \bsy{u}_2, \;  \dots, \; \bsy{u}_n \big) \in \mathbb{R}^{m \times n}$ where $\bsy{u}_j \in \mathbb{R}^m$ ($1 \leq j \leq n$), then $\supp(\bsy{U}) := \bigcup_{j \in \left[n \right]}  \supp(\bsy{u}_j)$. The definition above extends the notion of support to matrices, \textit{i.e.}, the support of a matrix is the union of the supports of its columns. Similarly to the vector case, if $\bsy{U} = \big( \bsy{u}_1, \;  \bsy{u}_2, \;  \dots, \; \bsy{u}_n \big) \in \mathbb{R}^{m \times n}$, then $
\|\bsy{U} \|_0 := |\supp(\bsy{U})| = \left|\bigcup_{j \in \left[n \right]} \supp(\bsy{u}_j) \right|$.\\

We define $\|\bsy{\phi}\|_{\mathrm{min}} := \min_{i \in \left[m \right]} |\phi_i|$ which is not a norm. 
%
Furthermore, for  $\bsy{\Phi} \in \mathbb{R}^{m \times n}$, $\|\bsy{\Phi} \|_{p \rightarrow q}$ is defined as \cite[Equation A.8]{foucart2013mathematical} $
\|\bsy{\Phi} \|_{p \rightarrow q} := \sup_{\|\bsy{\phi} \|_p = 1} \left\|\bsy{\Phi} \bsy{\phi}\right\|_q$ where $1 \leq p, q \leq \infty$.
For the sake of simplifying the notations, we will adopt the convention $\|\bsy{\Phi} \|_{p} := \|\bsy{\Phi} \|_{p \rightarrow p}$. It can be shown that, for $\bsy{\Phi} \in \mathbb{R}^{m \times n}$, \cite[Lemma A.5]{foucart2013mathematical}
\begin{equation}\label{eq:inftyToOneNorm}
\|\bsy{\Phi} \|_{\infty} = \max_{i \in \left[m \right]} \sum_{j=1}^{n} \left| \Phi_{i,j} \right| = \|\bsy{\Phi}^{\mathrm{T}} \|_{1}.  
\end{equation}
\subsection{Sparse rank}\label{sec:sparkDef}

The sparse rank \cite{donoho2003optimally} of $\bsy{\Phi} \in \mathbb{R}^{m \times n}$, denoted by $\mathrm{spark}(\bsy{\Phi})$, is given by 
\begin{equation}
\mathrm{spark}(\bsy{\Phi}) = \min_{\bsy{u} \neq \bsy{0}} \|\bsy{u} \|_0  \text{  s.t.  }  \bsy{\Phi} \bsy{u} = \bsy{0}. 
\end{equation}
$\mathrm{spark}(\bsy{\Phi})$ is thus the smallest number of linearly dependent columns of $\bsy{\Phi}$. Equivalently, it means that, if $\mathrm{spark}(\bsy{\Phi}) > s$, then, for every support $S$ such that $|S| \leq s$, the columns of $\bsy{\Phi}_{S}$ are linearly independent, \textit{i.e.}, $\bsy{\Phi}_{S}$ has full column rank.\\

Note that computing $\mathrm{spark}(\bsy{\Phi})$ for a given matrix $\bsy{\Phi}$ is not computationally tractable as this problem is even harder to solve than a $\ell_0$ norm minimization problem which is known to be NP-hard \cite{elad2010sparse}.

\subsection{Restricted isometry property}\label{subsec:RIPDef}

The matrix $\bsy{\Phi} \in \mathbb{R}^{m \times n}$ satisfies the so-called Restricted Isometry Property (RIP) \cite{candes2006stable} of order $s$ if there exists a constant $\delta_s < 1$ such that
\begin{equation}\label{eq:RIPDef}
(1 - \delta_s) \|\bsy{u}\|_2^2 \leq \|\bsy{\Phi} \bsy{u}\|_2^2 \leq (1 + \delta_s) \|\bsy{u}\|_2^2
\end{equation}
for all $s$-sparse vectors $\bsy{u}$. The smallest $\delta_s$ that satisfies Equation~(\ref{eq:RIPDef}) is called the Rectricted Isometry Constant (RIC) of order $s$. The RIC of order $s$ can theoretically be computed by considering
\begin{align*}
U_s = \max_{\substack{S \subseteq \left[n \right] \\ |S| = s}} \lambda_{\mathrm{max}}(\bsy{\Phi}_{S}^{\mathrm{T}} \bsy{\Phi}_{S}) - 1 \\
L_s = 1 - \min_{\substack{S \subseteq \left[n \right] \\ |S| = s}} \lambda_{\mathrm{min}}(\bsy{\Phi}_{S}^{\mathrm{T}} \bsy{\Phi}_{S})
\end{align*}
where $\lambda_{\mathrm{min}}$ and $\lambda_{\mathrm{max}}$ denote the smallest and largest eigenvalues respectively. Then,
\begin{equation}\label{eq:RICasMax}
\delta_s = \max(U_s, L_s).
\end{equation}

Evaluating $U_s$ and $L_s$ is not computationally tractable as it requires to determine the smallest and largest eigenvalues of $\binom{n}{s}$ matrices of size $m \times s$. In particular, this problem has been shown to be NP-Hard in the general case \cite{tillmann2013computational}.  It is therefore interesting to find an upper bound on $\delta_s$ that can be easily computed.

\begin{lemSA}[Coherence bound on the RIP]\label{lem:RICLambda}
If $\mu (s-1) < 1$, then 
\begin{align*}
\max_{\substack{S \subseteq \left[m \right] \\ |S| = s}} \lambda_{\mathrm{max}}(\bsy{\Phi}_{S}^{\mathrm{T}} \bsy{\Phi}_{S}) \leq 1 + \mu_1(s-1) \leq 1 + (s-1) \mu\\
\min_{\substack{S \subseteq \left[m \right] \\ |S| = s}} \lambda_{\mathrm{min}}(\bsy{\Phi}_{S}^{\mathrm{T}} \bsy{\Phi}_{S}) \geq 1 - \mu_1(s-1) \geq 1 - (s-1) \mu.
\end{align*}
The first inequality of each line holds if $\mu_1 (s-1) < 1$.
\end{lemSA}
\begin{proof}
This result is obtained in \cite{cai2011orthogonal}. 
\end{proof}%
A consequence of Lemma \ref{lem:RICLambda} is that, if $\mu (s-1) < 1$, then
\begin{equation}\label{eq:RICToCoherence}
\delta_s \leq \mu_1(s-1) \leq (s-1) \mu.
\end{equation}
It is worth noticing that if $\delta_s < 1$, then $\mathrm{spark}(\bsy{\Phi}) > s$. The reason is that $\delta_s < 1$ implies that $\min_{S \subseteq \left[m \right], |S| = s} \lambda_{\mathrm{min}}(\bsy{\Phi}_{S}^{\mathrm{T}} \bsy{\Phi}_{S}) > 0$ which in turn implies that, for every support $S$ of cardinality equal to or less than $s$, $\bsy{\Phi}_S$ has full column rank.

\subsection{Lipschitz functions}
A Lipschitz function $f : \mathbb{R}^K \rightarrow \mathbb{R} : \bsy{g} \mapsto f(\bsy{g})$ with regards to metric $\ell_2$ is a function that satisfies
\begin{equation}
\exists L > 0 : \forall \bsy{x}, \bsy{y} \in \mathbb{R}^K, \, \left| f(\bsy{x}) - f(\bsy{y}) \right| \leq L \left\| \bsy{x} - \bsy{y} \right\|_2 .
\end{equation}
The constant $L$ is called the Lipschitz constant.


\section{SOMP-NS: general theoretical results}\label{sec:section3}
We now wish to derive basic theoretical results needed to analyze the performance of SOMP-NS in both the noiseless and noisy cases. The theoretical framework developed hereafter will be used in Section~\ref{sec:recovGuarantees} to give lower bounds on the probability of partial or full recovery using SOMP-NS when the additive noise is Gaussian.

\subsection{Support-dependent exact recovery criterion without noise}\label{subsec:ERCNoiseless}

We now wish to build an exact recovery criterion (ERC) that guarantees, for dictionary matrix $\bsy{\Phi}$, the recovery of every sparse signal of support $S$ by means of SOMP-NS. The result that will be obtained is similar to an earlier result of J.~A.~Tropp, sometimes referred to as Tropp's ERC \cite{tropp2004greed}. \\

First of all, it is worth noticing that the maximum correlation in step 4 of Algorithm \ref{alg:SOMPNS} can be written as 
\begin{equation}
\mathrm{max}_{j \in \left[ n \right]} \left( \sum_{k=1}^{K} \left| \left\langle \bsy{r}^{(t)}_k, \bsy{\phi}_{j} \right\rangle  \right| q_k \right) = \left\|\bsy{\Phi}^{\mathrm{T}} \bsy{R}^{(t)} \bsy{Q} \right\|_{\infty}
\end{equation}
where $\bsy{Q} = \mathrm{diag}\left(q_1, \ldots, q_K \right)$ is the diagonal matrix that contains the weights to be applied to each column vector of $\bsy{R}^{(t)}$. It is assumed that $q_k \geq 0$ ($1 \leq k \leq K$).\\

The next result allows to state the ERC for SOMP-NS. It will also be used to develop the exact recovery criterion for the noisy case later.

\begin{lem}[A lower bound on the maximal residual projection]\label{lem:ERC}
Let $\bsy{\Phi} \in \mathbb{R}^{m \times n}$ where $\bsy{\Phi}_S$ has full column rank for some support $S \subset \lbrack n \rbrack$. Moreover, $\bsy{P}^{(t)} = \bsy{\Phi}_{S_t} \bsy{\Phi}_{S_t}^+$ denotes the orthogonal projector onto $\mathrm{span}(\bsy{\Phi}_{S_t})$ where  $S_t \subseteq S$.  Let $\bsy{R}^{(t)} = (\bsy{I} - \bsy{P}^{(t)}) \bsy{Y} = (\bsy{I} - \bsy{P}^{(t)}) \bsy{\Phi} \bsy{X}$, $\bsy{Y} \in \mathbb{R}^{m \times K}$ and $\bsy{X} \in \mathbb{R}^{n \times K}$ where $K \geq 1$. It is assumed that $\mathrm{supp}(\bsy{X}) = S$. Under these conditions, the following inequality holds true
\begin{equation}
\|\bsy{\Phi}_S^{\mathrm{T}} \bsy{R}^{(t)} \bsy{Q} \|_{\infty} \left\|\bsy{\Phi}_S^+ \bsy{\Phi}_{\overline{S}} \right\|_{1} \geq \|\bsy{\Phi}_{\overline{S}}^{\mathrm{T}} \bsy{R}^{(t)} \bsy{Q} \|_{\infty}
\end{equation}
where $\overline{S}$ denotes the relative complement of $S$ with respect to $\left[n \right]$.
\end{lem}

\begin{proof} 
Following the steps of \cite[Theorem 4.5]{chen2006theoretical}, the lemma is easily obtained. Each column of the matrix $\bsy{R}^{(t)}$ can be expressed as a linear combination of the columns of $\bsy{\Phi}_S$. The reason that explains this last statement is that $\bsy{Y}_k = (\bsy{\Phi} \bsy{X})_k = (\bsy{\Phi}_S \bsy{X}_S)_k \in \mathrm{span}(\bsy{\Phi}_S)$ and $((\bsy{I} - \bsy{P}^{(t)}) \bsy{Y})_k \in \mathrm{span}(\bsy{\Phi}_{S}) \cup \mathrm{span}(\bsy{\Phi}_{S_t}) = \mathrm{span}(\bsy{\Phi}_{S})$.\\

Moreover, $\bsy{\Phi}_S$ is guaranteed to have full column rank which implies that the Moore-Penrose pseudoinverse $\bsy{\Phi}_S^+$ is equal to $(\bsy{\Phi}_S^{\mathrm{T}} \bsy{\Phi}_S)^{-1} \bsy{\Phi}_S^{\mathrm{T}}$ and consequently that $ ((\bsy{\Phi}_S^+)^{\mathrm{T}} \bsy{\Phi}_S^{\mathrm{T}}) \bsy{R}^{(t)} =  \bsy{\Phi}_S \bsy{\Phi}_S^+ \bsy{R}^{(t)} = \bsy{R}^{(t)}$. Furthermore, it is easily established \cite[Lemma 4.4]{chen2006theoretical} that for two matrices $\bsy{A}$ and $\bsy{B}$, $\| \bsy{A} \bsy{B} \|_{\infty} \leq \| \bsy{A} \|_{\infty} \| \bsy{B} \|_{\infty}$.


Combining the results above yields
\begin{align*}
\|\bsy{\Phi}_{\overline{S}}^{\mathrm{T}} \bsy{R}^{(t)}  \bsy{Q} \|_{\infty} &= \|\bsy{\Phi}_{\overline{S}}^{\mathrm{T}} (\bsy{\Phi}_S^+)^{\mathrm{T}} \bsy{\Phi}_S^{\mathrm{T}} \bsy{R}^{(t)}  \bsy{Q} \|_{\infty} \\
 & \leq \|\bsy{\Phi}_{\overline{S}}^{\mathrm{T}} (\bsy{\Phi}_S^+)^{\mathrm{T}} \|_{\infty} \| \bsy{\Phi}_S^{\mathrm{T}} \bsy{R}^{(t)}  \bsy{Q} \|_{\infty}
\end{align*}
Finally, using Equation~(\ref{eq:inftyToOneNorm}) shows that $\| \bsy{\Phi}_{\overline{S}}^{\mathrm{T}} (\bsy{\Phi}_S^+)^{\mathrm{T}} \|_{\infty} = \left\|\bsy{\Phi}_S^+ \bsy{\Phi}_{\overline{S}} \right\|_{1}$ and concludes the proof. \end{proof}

We are now ready to provide an ERC for SOMP-NS.

\begin{thm}[ERC for the noiseless SOMP-NS algorithm]\label{thm:ERC}
Let $\bsy{\Phi} \in \mathbb{R}^{m \times n}$ and $\bsy{X} \in \mathbb{R}^{n \times K}$. If $\bsy{Y} = \bsy{\Phi} \bsy{X}$ and $\bsy{\Phi}_S$ has full column rank, then  a sufficient condition for SOMP-NS to properly retrieve the support $S$ of $\bsy{X}$ after exactly $|S|$ iterations is
\begin{equation}
\left\|\bsy{\Phi}_S^+ \bsy{\Phi}_{\overline{S}} \right\|_{1} < 1
\end{equation}
where $\overline{S}$ is the relative complement of $S$ with respect to $\left[n \right]$.
\end{thm}
\begin{proof} Let us assume that SOMP-NS has made correct decisions before iteration $t$. The greedy selection ratio (at iteration $t$) is defined as:
\begin{equation} \label{eq:greedySelecRatioDef}
\rho_t = \frac{\|\bsy{\Phi}_{\overline{S}}^{\mathrm{T}} \bsy{R}^{(t)} \bsy{Q} \|_{\infty}}{\|\bsy{\Phi}_{S}^{\mathrm{T}} \bsy{R}^{(t)} \bsy{Q} \|_{\infty}}.
\end{equation}
Clearly, $\rho_t < 1$ ensures that SOMP-NS performs a correct decision at iteration $t$ since it implies that the largest sum of inner products is obtained for one of the atoms belonging to the support $S$, \textit{i.e.}, $\|\bsy{\Phi}_{S}^{\mathrm{T}} \bsy{R}^{(t)} \bsy{Q} \|_{\infty} > \|\bsy{\Phi}_{\overline{S}}^{\mathrm{T}} \bsy{R}^{(t)} \bsy{Q} \|_{\infty}$. Since it is assumed that SOMP-NS has only made correct decisions so far, Lemma \ref{lem:ERC} shows that $\rho_t < 1$ holds true whenever $\left\|\bsy{\Phi}_S^+ \bsy{\Phi}_{\overline{S}} \right\|_{1} < 1$.\\

Furthermore, $\left\|\bsy{\Phi}_S^+ \bsy{\Phi}_{\overline{S}} \right\|_{1} < 1$ ensures that a correct atom is picked during iteration $0$. By induction, the hypothesis that correct decisions have been made so far will be satisfied for every iteration and full support recovery at iteration $|S|-1$ is thus guaranteed.
\end{proof}

It can be shown that Theorem~\ref{thm:ERC} is sharp in the sense that if the inequality $\left\|\bsy{\Phi}_S^+ \bsy{\Phi}_{\overline{S}} \right\|_{1} < 1$  is not satisfied for some support $S$, it is always possible to find a $\bsy{X}_{\mathrm{bad}}$ whose support is $S$ for which SOMP-NS identifies an incorrect atom at the first iteration. The sharpness property for SOMP-NS directly derives from \cite[Theorem 3.10]{tropp2004greed} which provides an equivalent property in the SMV case for OMP. Indeed, if OMP fails to choose a correct atom with vector $\bsy{x}_{\mathrm{bad}}$, then SOMP-NS also fails with $\bsy{X}_{\mathrm{bad}} = \big(\bsy{x}_{\mathrm{bad}},\; \dots, \; \bsy{x}_{\mathrm{bad}}\big)$ as, in this particular case, OMP with signal $\bsy{y} = \bsy{\Phi} \bsy{x}_{\mathrm{bad}}$ and SOMP-NS with signal $\bsy{Y} = \bsy{\Phi} \bsy{X}_{\mathrm{bad}}$ make the same decisions.\\

Moreover, the theorem above is generally of theoretical use only since computing $\|\bsy{\Phi}_S^+ \bsy{\Phi}_{\overline{S}} \|_{1}$ requires to know the support beforehand. Also, computing $\|\bsy{\Phi}_S^+ \bsy{\Phi}_{\overline{S}} \|_{1}$ for all the possible supports of a given size is not computationally tractable.\\

However, it is shown in \cite{tropp2004greed} that if $\mu < \frac{1}{2 |S| -1}$, then $ \|\bsy{\Phi}_S^+ \bsy{\Phi}_{\overline{S}} \|_{1} \leq \frac{|S| \mu}{1 - (|S|-1) \mu} < 1$ for all the supports of size $|S|$. It is worth noticing that $\bsy{\Phi}_{S}$ is full rank for all supports $S$ of size $s$ if and only if $\mathrm{spark}(\bsy{\Phi}) > s$.\\

%

Although Theorem~\ref{thm:ERC} proves to be interesting for $\bsy{E} = \bsy{0}$ in signal model (\ref{eq:MMVSignalModel}), it should be extended to the noisy case which is the purpose of the next section.

\subsection{Correct support detection criterion in the noisy case}

We develop here an ERC for the noisy case which generalizes Theorem~\ref{thm:ERC} to this context. To reach this result, let us assume that SOMP-NS has made correct decisions before iteration $t$. First of all, we separate the contribution of the noise $\bsy{E}$ and that of the useful signal $\bsy{X}$:
\begin{align*}
\bsy{R}^{(t)}  &= (\bsy{I} - \bsy{P}^{(t)})\bsy{Y} \\
 &= (\bsy{I} - \bsy{P}^{(t)})(\bsy{\Phi} \bsy{X} + \bsy{E}) \\
 &= \underbrace{(\bsy{I} - \bsy{P}^{(t)} ) \bsy{\Phi} \bsy{X}}_{ = \bsy{Z}^{(t)} }  + \underbrace{(\bsy{I} - \bsy{P}^{(t)} ) \bsy{E}}_{ = \bsy{E}^{(t)} }\\
 &= \bsy{Z}^{(t)}  + \bsy{E}^{(t)}.
\end{align*}
Then, the SOMP-NS correlation is lower bounded by
\begin{equation} \label{eq:firstBigApprox}
\begin{aligned}
\|\bsy{\Phi}_S^{\mathrm{T}} \bsy{R}^{(t)} \bsy{Q} \|_{\infty} & = \|\bsy{\Phi}_S^{\mathrm{T}} (\bsy{Z}^{(t)} + \bsy{E}^{(t)}) \bsy{Q} \|_{\infty} \\
 & \geq \|\bsy{\Phi}_S^{\mathrm{T}} \bsy{Z}^{(t)} \bsy{Q} \|_{\infty} - \|\bsy{\Phi}_S^{\mathrm{T}} \bsy{E}^{(t)} \bsy{Q} \|_{\infty}.
\end{aligned}
\end{equation}
%
%
Moreover, the triangle inequality yields
\begin{equation}\label{eq:secondBigApprox}
\|\bsy{\Phi}_{\overline{S}}^{\mathrm{T}} \bsy{R}^{(t)} \bsy{Q} \|_{\infty} \leq \|\bsy{\Phi}_{\overline{S}}^{\mathrm{T}} \bsy{Z}^{(t)} \bsy{Q} \|_{\infty} + \|\bsy{\Phi}_{\overline{S}}^{\mathrm{T}} \bsy{E}^{(t)} \bsy{Q} \|_{\infty}.
\end{equation}
Since SOMP-NS makes a correct decision at step $t$ if $\|\bsy{\Phi}_S^{\mathrm{T}} \bsy{R}^{(t)} \bsy{Q} \|_{\infty} > \|\bsy{\Phi}_{\overline{S}}^{\mathrm{T}} \bsy{R}^{(t)} \bsy{Q} \|_{\infty}$, the inequalities above show that this condition is always satisfied whenever
\begin{equation}\label{eq:SOMPanalysis1}
\begin{aligned}  \|\bsy{\Phi}_S^{\mathrm{T}} \bsy{Z}^{(t)} \bsy{Q} \|_{\infty} - \|\bsy{\Phi}_{\overline{S}}^{\mathrm{T}} \bsy{Z}^{(t)} \bsy{Q} \|_{\infty} >  \|\bsy{\Phi}_S^{\mathrm{T}} \bsy{E}^{(t)} \bsy{Q} \|_{\infty} + \|\bsy{\Phi}_{\overline{S}}^{\mathrm{T}} \bsy{E}^{(t)} \bsy{Q} \|_{\infty}. \end{aligned}
\end{equation}
It is worth noticing that, due to Lemma \ref{lem:ERC}, the following relationship holds true:
\begin{equation*}
\|\bsy{\Phi}_{\overline{S}}^{\mathrm{T}} \bsy{Z}^{(t)} \bsy{Q} \|_{\infty} \leq \|\bsy{\Phi}_S^{\mathrm{T}} \bsy{Z}^{(t)} \bsy{Q} \|_{\infty} \|\bsy{\Phi}_S^+ \bsy{\Phi}_{\overline{S}} \|_{1}.
\end{equation*}
Note that the assumption that correct decisions have been so far is crucial for this result to be true. Furthermore, as $\bsy{\Phi}_S$ and $\bsy{\Phi}_{\overline{S}}$ are column submatrices of $\bsy{\Phi}$, one easily obtains
\begin{align}
\|\bsy{\Phi}_S^{\mathrm{T}} \bsy{E}^{(t)} \bsy{Q} \|_{\infty} \leq \|\bsy{\Phi}^{\mathrm{T}} \bsy{E}^{(t)} \bsy{Q} \|_{\infty}, \\
\|\bsy{\Phi}_{\overline{S}}^{\mathrm{T}} \bsy{E}^{(t)} \bsy{Q} \|_{\infty} \leq \|\bsy{\Phi}^{\mathrm{T}} \bsy{E}^{(t)} \bsy{Q} \|_{\infty} .
\end{align}
A sufficient condition for (\ref{eq:SOMPanalysis1}) to hold is thereby
\begin{equation*}
\left(1 - \|\bsy{\Phi}_S^+ \bsy{\Phi}_{\overline{S}} \|_{1} \right) \|\bsy{\Phi}_S^{\mathrm{T}} \bsy{Z}^{(t)} \bsy{Q} \|_{\infty} > 2 \|\bsy{\Phi}^{\mathrm{T}} \bsy{E}^{(t)} \bsy{Q} \|_{\infty}.
\end{equation*}
Theorem~\ref{thm:TheoFramework1} summarizes the previous discussion by explicitly stating the ERC in the noisy case.

\begin{thm}[ERC for the noisy SOMP-NS algorithm]\label{thm:TheoFramework1}
Let $\bsy{X}$ be the sparse matrix to be retrieved and let $\mathrm{supp}(\bsy{X}) = S$ denote its support. Let $\bsy{\Phi}_S$ exhibit a full column rank. Let us assume that only correct atoms have been picked before iteration $t < |S|$ and that the reduced dictionary matrix $\bsy{\Phi}_S$ has full column rank. SOMP-NS with dictionary matrix $\bsy{\Phi}$ and signal $\bsy{Y} = \bsy{\Phi} \bsy{X} + \bsy{E}$ is guaranteed to make a correct decision at iteration $t$ whenever
\begin{equation}\label{eq:SOMPanalysis2}
\left(1 - \|\bsy{\Phi}_S^+ \bsy{\Phi}_{\overline{S}} \|_{1} \right) \|\bsy{\Phi}_S^{\mathrm{T}} \bsy{Z}^{(t)} \bsy{Q} \|_{\infty} > 2 \|\bsy{\Phi}^{\mathrm{T}} \bsy{E}^{(t)} \bsy{Q} \|_{\infty}
\end{equation}
where $\bsy{R}^{(t)} = (\bsy{I} - \bsy{P}^{(t)})\bsy{Y}$, $\bsy{Z}^{(t)} = (\bsy{I} - \bsy{P}^{(t)} ) \bsy{\Phi} \bsy{X}$ and $\bsy{E}^{(t)} = (\bsy{I} - \bsy{P}^{(t)} ) \bsy{E}$.
\end{thm}
Noticeably a necessary condition for satisfying (\ref{eq:SOMPanalysis2}) reads $\|\bsy{\Phi}_S^+ \bsy{\Phi}_{\overline{S}} \|_{1} < 1$ which is precisely the ERC obtained before for the noiseless case. Moreover, low values of $\|\bsy{\Phi}_S^+ \bsy{\Phi}_{\overline{S}} \|_{1}$ imply a better robustness against the noise as the condition hereabove is then more easily satisfied. This observation is not surprising since the robustness in the noiseless case determines the amplitude of the noise we can apply without ruining the support recovery.\\

Theorem~\ref{thm:TheoFramework1} is the cornerstone of the theoretical analysis conducted in this paper. However, as done hereafter, it remains desirable to find a lower bound for $\|\bsy{\Phi}_S^{\mathrm{T}} \bsy{Z}^{(t)} \bsy{Q} \|_{\infty}$ that expresses in a simpler manner the impact of the signal to be estimated, the weights and the dictionary matrix. Also, the term $\|\bsy{\Phi}^{\mathrm{T}} \bsy{E}^{(t)} \bsy{Q} \|_{\infty}$ will be dealt with by means of a statistical analysis presented in Section~\ref{sec:recovGuarantees} when the noise is Gaussian.

\subsection{Lower bound on $\|\bsy{\Phi}_{S}^{\mathrm{T}} \bsy{Z}^{(t)} \bsy{Q} \|_{\infty}$}

We now focus on deriving a lower bound for $\|\bsy{\Phi}_{S}^{\mathrm{T}} \bsy{Z}^{(t)} \bsy{Q} \|_{\infty}$ that can be easily evaluated. Verifying Equation~(\ref{eq:SOMPanalysis2}) with this computable lower bound will provide a sufficient condition for the ERC to hold. In particular, we desire to obtain a lower bound that does not rely on the knowledge of the particular support $S$ that is chosen.\\

To reach this goal, we chose here to extend the method proposed in \cite{cai2011orthogonal} to MMV problems. The following theorem mainly has a theoretical interest and is afterwards particularized in Corollary \ref{corr:RIPBasedBtBound} that will used in the rest of the paper. Corollary~\ref{corr:muBasedBtBound} will provide a variant of Corollary~\ref{corr:RIPBasedBtBound} relying on both the coherence of the dictionary (instead of the RIP) and on the quantity $\|\bsy{\Phi}_S^+ \bsy{\Phi}_{\overline{S}} \|_{1}$.\\

\begin{thm}\label{thm:firstBoundGeneral}
Let $S := \mathrm{supp}(\bsy{X}) \subset \left[n \right]$ and let $S_t$ denote the indexes of the atoms chosen by SOMP-NS at iteration $t$. It is assumed that  $S_t \subset S$, \textit{i.e.}, only correct decisions have been made before iteration $t$. $\mathcal{J}_t = S \backslash S_t$ contains the indexes of the correct atoms yet to be selected at iteration $t$. Let $\bsy{Z}^{(t)} = (\bsy{I} - \bsy{P}^{(t)}) \bsy{\Phi} \bsy{X}$ where $\bsy{P}^{(t)} = \bsy{\Phi}_{S_t} \bsy{\Phi}_{S_t}^+$ denotes the orthogonal projector onto $\mathrm{span}(\bsy{\Phi}_{S_t})$. Then, for any $c_k^{(t)} \in \lbrace -1, 1 \rbrace$ ($1 \leq k \leq K$), 
\begin{equation}\label{eq:firstBound}
\|\bsy{\Phi}_S^{\mathrm{T}} \bsy{Z}^{(t)} \bsy{Q} \|_{\infty} \geq  \left|\mathcal{J}_t\right|^{-1/2} \lambda_{\mathrm{min}}\left(\bsy{\Phi}_{\mathcal{J}_t}^{\mathrm{T}}  \bsy{\Phi}_{\mathcal{J}_t}\right)  \left\| \overline{\bsy{x}}^{(t)}_{\mathcal{J}_t} \right\|_{2}
\end{equation}
where $\overline{\bsy{x}}^{(t)} = \sum_{k=1}^{K} \bsy{x}_k c_k^{(t)} q_k$. Moreover, if $\bsy{\Phi}$ satisfies the RIP with $|\mathcal{J}_t|$-th restricted isometry constant $\delta_{|\mathcal{J}_t|}  < 1$, then
\begin{equation}\label{eq:firstBoundv2}
\|\bsy{\Phi}_S^{\mathrm{T}} \bsy{Z}^{(t)} \bsy{Q} \|_{\infty} \geq  \left|\mathcal{J}_t\right|^{-1/2} (1-\delta_{|\mathcal{J}_t|})  \left\|  \overline{\bsy{x}}^{(t)}_{\mathcal{J}_t}  \right\|_{2}.
\end{equation}
\end{thm}

\begin{proof} Denoting the $k$-th column of $\bsy{Z}^{(t)}$ by $\bsy{z}^{(t)}_k$, we first observe that
\begin{equation}
\|\bsy{\Phi}_S^{\mathrm{T}} \bsy{Z}^{(t)} \bsy{Q} \|_{\infty} = \max_{j \in \mathcal{J}_t} \left( \sum_{k=1}^{K} \left| \left\langle \bsy{\phi}_j, \bsy{z}^{(t)}_k \right\rangle \right| q_k \right),
\end{equation}
the maximum being taken over $j \in \mathcal{J}_t$ since $\bsy{z}^{(t)}_k$ is orthogonal to $\mathrm{span}(\bsy{\Phi}_{S_t})$ because of the orthogonal projector $(\bsy{I} - \bsy{P}^{(t)})$. Since $|c_k^{(t)}| = 1$, $ |\langle \bsy{\phi}_j, \bsy{z}^{(t)}_k c_k^{(t)} \rangle | = | \langle \bsy{\phi}_j, \bsy{z}^{(t)}_k \rangle c_k^{(t)} | = | \langle \bsy{\phi}_j, \bsy{z}^{(t)}_k \rangle |$ which implies
\begin{equation}
\|\bsy{\Phi}_S^{\mathrm{T}} \bsy{Z}^{(t)} \bsy{Q} \|_{\infty} = \max_{j \in \mathcal{J}_t} \left( \sum_{k=1}^{K} \left| \left\langle \bsy{\phi}_j, \bsy{z}^{(t)}_k c_k^{(t)} q_k \right\rangle \right| \right).
\end{equation}
The triangle inequality yields
\begin{align*}
\|\bsy{\Phi}_S^{\mathrm{T}} \bsy{Z}^{(t)} \bsy{Q} \|_{\infty} & \geq \max_{j \in \mathcal{J}_t} \left( \left| \sum_{k=1}^{K} \left\langle \bsy{\phi}_j, \bsy{z}^{(t)}_k c_k^{(t)} q_k \right\rangle \right|  \right) \\
 & =\left\| \bsy{\Phi}^{\mathrm{T}}_{\mathcal{J}_t} (\bsy{I} - \bsy{P}^{(t)}) \bsy{\Phi}_{\mathcal{J}_t} \overline{\bsy{x}}^{(t)}_{\mathcal{J}_t} \right\|_{\infty}.
\end{align*}
%
Thus,
\begin{align*}
\|\bsy{\Phi}_S^{\mathrm{T}} \bsy{Z}^{(t)} \bsy{Q} \|_{\infty} &  \geq \left|\mathcal{J}_t\right|^{-1/2} \left\| \vphantom{\sum_{k=1}^{K}} \bsy{\Phi}^{\mathrm{T}}_{\mathcal{J}_t} (\bsy{I} - \bsy{P}^{(t)}) \bsy{\Phi}_{\mathcal{J}_t} \overline{\bsy{x}}^{(t)}_{\mathcal{J}_t} \right\|_{2}  \\
  &  \geq  \left|\mathcal{J}_t \right|^{-1/2} \lambda_{\mathrm{min}}\left(\bsy{\Phi}_{\mathcal{J}_t}^{\mathrm{T}} (\bsy{I} - \bsy{P}^{(t)}) \bsy{\Phi}_{\mathcal{J}_t}\right) \left\| \overline{\bsy{x}}^{(t)}_{\mathcal{J}_t}  \right\|_{2}\\
    &\geq  \left|\mathcal{J}_t\right|^{-1/2} \lambda_{\mathrm{min}}\left(\bsy{\Phi}_{\mathcal{J}_t}^{\mathrm{T}}  \bsy{\Phi}_{\mathcal{J}_t}\right)  \left\|  \overline{\bsy{x}}^{(t)}_{\mathcal{J}_t}  \right\|_{2}.
\end{align*}
The first inequality results from the observation that, for any vector $\bsy{x} \in \mathbb{R}^{|\mathcal{J}_t|}$, we have $ \| \bsy{x} \|_2 \leq \sqrt{|\mathcal{J}_t|} \| \bsy{x} \|_{\infty}$. The inequality $\lambda_{\mathrm{min}}(\bsy{\Phi}_{\mathcal{J}_t}^{\mathrm{T}} (\bsy{I} - \bsy{P}^{(t)}) \bsy{\Phi}_{\mathcal{J}_t}) \geq \lambda_{\mathrm{min}}(\bsy{\Phi}_{\mathcal{J}_t}^{\mathrm{T}}  \bsy{\Phi}_{\mathcal{J}_t})$ is available in \cite[Lemma 5]{cai2011orthogonal}.  The first part of the theorem is now proved.\\

If $\bsy{\Phi}$ satisfies the RIP with RIC $\delta_{|\mathcal{J}_t|}  < 1$, then Equation~(\ref{eq:RICasMax}) yields $1- \delta_{|\mathcal{J}_t|} \leq \lambda_{\mathrm{min}}\left(\bsy{\Phi}_{\mathcal{J}_t}^{\mathrm{T}}  \bsy{\Phi}_{\mathcal{J}_t}\right)$, which proves the second part of the theorem. \end{proof}

The result above shows that the decision metric in \mbox{SOMP-NS} corresponding to the correct atoms, \textit{i.e.}, $\|\bsy{\Phi}_S^{\mathrm{T}} \bsy{Z}^{(t)} \bsy{Q} \|_{\infty}$, is closely related to the $\ell_2$ norm of the signal to be recovered and to the singular values of $\bsy{\Phi}_{\mathcal{J}_t}$. It is clear that the ability of $\bsy{\Phi}$ to conserve the norm of sparse vectors is necessary to ensure that the measurement noise $\bsy{E}^{(t)}$ does not absorb $\bsy{Z}^{(t)}$.\\

Through the following corollary, we wish to obtain a simple term that replaces $\|  \overline{\bsy{x}}^{(t)}_{\mathcal{J}_t}  \|_{2}$.

\begin{corr}\label{corr:RIPBasedBtBound}
Let $S := \mathrm{supp}(\bsy{X}) \subset \left[n \right]$ and let $S_t$ denote the indexes of the atoms chosen by SOMP-NS at iteration $t$. It is assumed that  $S_t \subset S$, \textit{i.e.}, only correct decisions have been made before iteration $t$. $\mathcal{J}_t$ contains the indexes of the correct atoms yet to be selected at iteration $t$. Let $\bsy{Z}^{(t)} = (\bsy{I} - \bsy{P}^{(t)}) \bsy{\Phi} \bsy{X}$ where $\bsy{P}^{(t)} = \bsy{\Phi}_{S_t} \bsy{\Phi}_{S_t}^+$ denotes the orthogonal projector onto $\mathrm{span}(\bsy{\Phi}_{S_t})$. If $\bsy{\Phi}$ satisfies the RIP with $|\mathcal{J}_t|$-th restricted isometry constant $\delta_{|\mathcal{J}_t|}  < 1$, then
\begin{align}
\|\bsy{\Phi}_S^{\mathrm{T}} \bsy{Z}^{(t)} \bsy{Q} \|_{\infty}  & \geq  (1-\delta_{|\mathcal{J}_t|})  \min_{j \in S} \sum_{k=1}^{K} | X_{j,k} | q_k \label{eq:firstBoundCorr1_2}\\
 & \geq \mathrm{tr}(\bsy{Q})  (1-\delta_{|\mathcal{J}_t|})  \min_{\substack{j \in S \\ k \in \lbrack K \rbrack}} |X_{j,k}|. \label{eq:firstBoundCorr1}
\end{align}
\end{corr}

\begin{proof}
It is easy to show that for any $\bsy{x} \in \mathbb{R}^{n}$
\begin{equation*}
\left\| \bsy{x}_{\mathcal{J}_t} \right\|_{2} \geq \sqrt{|\mathcal{J}_t|} \left\| \bsy{x}_{\mathcal{J}_t} \right\|_{\mathrm{min}} \geq \sqrt{|\mathcal{J}_t|} \left\| \bsy{x} \right\|_{\mathrm{min}}.
\end{equation*}
It implies that 
\begin{equation*}
\left\|  \sum_{k=1}^{K} (\bsy{x}_k)_{\mathcal{J}_t} c_k^{(t)} q_k  \right\|_{2} \geq \sqrt{|\mathcal{J}_t|} \min_{j \in S} \sum_{k=1}^{K} X_{j,k} c_k^{(t)} q_k.
\end{equation*}
In particular, the choice of the $c_k^{(t)}$ is arbitrary so that the best lower bound is given by
\begin{equation*}
\left\|  \sum_{k=1}^{K} (\bsy{x}_k)_{\mathcal{J}_t} c_k^{(t)} q_k  \right\|_{2} \geq \sqrt{|\mathcal{J}_t|} \max_{\lbrace c_k^{(t)} \rbrace} \min_{j \in S} \sum_{k=1}^{K} X_{j,k} c_k^{(t)} q_k.
\end{equation*}
One easily notices that $\max_{\lbrace c_k^{(t)} \rbrace} \min_{j \in S} \sum_{k=1}^{K} X_{j,k} c_k^{(t)} q_k \leq \min_{j \in S} \sum_{k=1}^{K} |X_{j,k}| q_k$. If $j^{*} := \argmin_j \sum_{k=1}^{K} |X_{j,k}| q_k$, then choosing $c_k^{(t)} = \mathrm{sign} (X_{j^{*}, k})$ is optimal, \newline \textit{i.e.}, $\max_{\lbrace c_k^{(t)} \rbrace} \min_{j \in S} \sum_{k=1}^{K} X_{j,k} c_k^{(t)} q_k =  \min_{j \in S} \sum_{k=1}^{K} |X_{j,k}| q_k = \sum_{k=1}^{K} |X_{j^{*},k}| q_k $. Together with Theorem~\ref{thm:firstBoundGeneral}, this result shows that the first inequality holds true. The last inequality is then trivially obtained. \end{proof}

Both inequalities in the result above explicitly emphasize the impact of the weights, the RIC and the amplitude of the coefficients to be recovered. However, regarding the weights, it is clear that they will also impact the value of $\|\bsy{\Phi}^{\mathrm{T}} \bsy{E}^{(t)} \bsy{Q} \|_{\infty}$.  Moreover, Equation~(\ref{eq:firstBoundCorr1_2}) suggests that, in order to reliably retrieve an atom, the sum of the absolute values of the associated coefficients in $\bsy{X}$ should be high enough.\\

As already mentioned in the introduction, Theorem~\ref{thm:firstBoundGeneral} appears to be a new result while Corollary \ref{corr:RIPBasedBtBound} has already been obtained in the literature \cite[Theorem 5]{gribonval2008atoms}.\\

It is now possible to use the coherence-based inequalities provided by Lemma \ref{lem:RICLambda} so as to derive a new bound on the basis of the previous one. We obtain Corollary \ref{corr:muBasedBtBound} which should be understood as the coherence counterpart of Corollary \ref{corr:RIPBasedBtBound}.

\begin{corr}\label{corr:muBasedBtBound}
Let $S := \mathrm{supp}(\bsy{X}) \subset \left[n \right]$ and let $S_t$ denote the indexes of the atoms chosen by SOMP-NS at iteration $t$. It is assumed that  $S_t \subset S$, \textit{i.e.}, only correct decisions have been made before iteration $t$. $\mathcal{J}_t$ contains the indexes of the correct atoms yet to be selected at iteration $t$. Let $\bsy{Z}^{(t)} = (\bsy{I} - \bsy{P}^{(t)}) \bsy{\Phi} \bsy{X}$ where $\bsy{P}^{(t)} = \bsy{\Phi}_{S_t} \bsy{\Phi}_{S_t}^+$ denotes the orthogonal projector onto $\mathrm{span}(\bsy{\Phi}_{S_t})$. If $\mu_1(|\mathcal{J}_t|-1) < 1$, then
\begin{equation} \label{eq:firstBoundvmu1}
\|\bsy{\Phi}_S^{\mathrm{T}} \bsy{Z}^{(t)} \bsy{Q} \|_{\infty}  \geq  (1 - \mu_1(|\mathcal{J}_t|-1))  \min_{j \in S} \sum_{k=1}^{K} | X_{j,k} | q_k.
\end{equation}
Moreover, if $(|S|-t-1)\mu < 1$, then both Equation~(\ref{eq:firstBoundvmu1}) and (\ref{eq:firstBoundvmu}) hold true.
\begin{equation}\label{eq:firstBoundvmu}
\|\bsy{\Phi}_S^{\mathrm{T}} \bsy{Z}^{(t)} \bsy{Q} \|_{\infty} \geq  (1 - (|S|-t-1)\mu)  \min_{j \in S} \sum_{k=1}^{K} | X_{j,k} | q_k. 
\end{equation}
\end{corr}

\begin{proof}
Using Lemma \ref{lem:RICLambda} and the equality $|\mathcal{J}_t| = |S| - t$, one obtains
\begin{align*}
(1-\delta_{|\mathcal{J}_t|}) & \geq 1 - \mu_1(|\mathcal{J}_t|-1) \\
 & \geq 1 - (|\mathcal{J}_t|-1) \mu \\
 & = 1 - (|S|-t-1)\mu.
\end{align*}
The first inequality makes sense only if $\mu_1(|\mathcal{J}_t|-1) < 1$ while the last inequality requires $(|S|-t-1)\mu < 1$.
\end{proof}

Although Corollary~\ref{corr:muBasedBtBound} is less powerful and general than Corollary~\ref{corr:RIPBasedBtBound}, it provides an interesting insight into how the coherence of the dictionary influences $\|\bsy{\Phi}_S^{\mathrm{T}} \bsy{Z}^{(t)} \bsy{Q} \|_{\infty}$.

\section{Theoretical analysis for the Gaussian noise case} \label{sec:recovGuarantees}

The previous section provided a non-probabilistic analysis of the quantity $\|\bsy{\Phi}_S^{\mathrm{T}} \bsy{Z}^{(t)} \bsy{Q} \|_{\infty}$ by deriving lower bounds that are more simple to evaluate than the original quantity. Regarding the noise-related quantity $\|\bsy{\Phi}^{\mathrm{T}} \bsy{E}^{(t)} \bsy{Q} \|_{\infty}$, we will in this section perform a stochastic analysis to derive a lower bound on the probability that it does not exceed a threshold $\varepsilon$ for Gaussian noises.\\

As shown by Theorem~\ref{thm:TheoFramework1}, it is possible to examine whether SOMP-NS succeeds in choosing a correct atom at step $t$ by evaluating separately quantities linked to the sparse signal to be estimated and the noise vectors, $\|\bsy{\Phi}_S^{\mathrm{T}} \bsy{Z}^{(t)} \bsy{Q} \|_{\infty}$ and $\|\bsy{\Phi}^{\mathrm{T}} \bsy{E}^{(t)} \bsy{Q} \|_{\infty}$ respectively. Since several simple lower bounds for $\|\bsy{\Phi}_S^{\mathrm{T}} \bsy{Z}^{(t)} \bsy{Q} \|_{\infty}$ have been found, it becomes possible to evaluate a lower bound on
\begin{equation}
\mathbb{P}\left[ \|\bsy{\Phi}^{\mathrm{T}} \bsy{E}^{(t)} \bsy{Q} \|_{\infty} < 0.5 \left(1 - \|\bsy{\Phi}_S^+ \bsy{\Phi}_{\overline{S}} \|_{1} \right) \|\bsy{\Phi}_S^{\mathrm{T}} \bsy{Z}^{(t)} \bsy{Q} \|_{\infty}  \right]
\end{equation}
\textit{i.e.}, a lower bound on the probability that SOMP-NS makes correct decisions for signal model (\ref{eq:MMVSignalModel}) according to Equation~(\ref{eq:SOMPanalysis2}) of Theorem~\ref{thm:TheoFramework1}. Corollaries~\ref{corr:RIPBasedBtBound} and \ref{corr:muBasedBtBound} yield
\begin{equation}\label{eq:ProbInequalities}
\begin{aligned}
\mathbb{P}\left[ \|\bsy{\Phi}^{\mathrm{T}} \bsy{E}^{(t)} \bsy{Q} \|_{\infty} < 0.5 \left(1 - \|\bsy{\Phi}_S^+ \bsy{\Phi}_{\overline{S}} \|_{1} \right) \|\bsy{\Phi}_S^{\mathrm{T}} \bsy{Z}^{(t)} \bsy{Q} \|_{\infty}  \right] \\
\geq \mathbb{P}\left[ \|\bsy{\Phi}^{\mathrm{T}} \bsy{E}^{(t)} \bsy{Q} \|_{\infty} < 0.5 \left(1 - \|\bsy{\Phi}_S^+ \bsy{\Phi}_{\overline{S}} \|_{1} \right) \|\bsy{\Phi}_S^{\mathrm{T}} \bsy{Z}^{(t)} \bsy{Q} \|_{\infty}^{\mathrm{(RIP)}}  \right] \\
\geq \mathbb{P}\left[ \|\bsy{\Phi}^{\mathrm{T}} \bsy{E}^{(t)} \bsy{Q} \|_{\infty} < 0.5 \left(1 - \|\bsy{\Phi}_S^+ \bsy{\Phi}_{\overline{S}} \|_{1} \right) \|\bsy{\Phi}_S^{\mathrm{T}} \bsy{Z}^{(t)} \bsy{Q} \|_{\infty}^{(\mu)}  \right]
\end{aligned}
\end{equation}
where
%
\begin{align}
 & \|\bsy{\Phi}_S^{\mathrm{T}} \bsy{Z}^{(t)} \bsy{Q} \|_{\infty}^{\mathrm{(RIP)}} :=  (1-\delta_{|\mathcal{J}_t|})  \min_{j \in S} \sum_{k=1}^{K} | X_{j,k} | q_k \label{eq:EqRIPBound} \\
 & \|\bsy{\Phi}_S^{\mathrm{T}} \bsy{Z}^{(t)} \bsy{Q} \|_{\infty}^{(\mu)} := \left(1 - (|S|-t-1)\mu\right)  \min_{j \in S} \sum_{k=1}^{K} | X_{j,k} | q_k. \label{eq:EqMuBound}
\end{align}

A statistical analysis of $\|\bsy{\Phi}^{\mathrm{T}} \bsy{E}^{(t)} \bsy{Q} \|_{\infty}$ is proposed when $\bsy{e}_k \sim \mathcal{N}(0, \sigma_k^2 \bsy{I}_{m \times m})$ and $\bsy{e}_{k_1} \independent \bsy{e}_{k_2}$ for $k_1 \neq k_2$. The advantage of our approach is to take into account the isotropic nature of statistically independent Gaussian random vectors. Our main result is Theorem~\ref{thm:finalGaussianThm}, which shows that the probability of making incorrect decisions from iteration $0$ to iteration $s < |S|$ included decreases exponentially with regards to a certain number of parameters. This theorem is then particularized so as to make use of the coherence $\mu$ of $\bsy{\Phi}$ instead of the RIP and the ERC.\\

First of all, the statistical properties of $\sum_{k=1}^{K} \left| \left\langle \bsy{\phi}_j, \bsy{e}_{k}  \right\rangle \right| q_k$ for a single and arbitrary atom $\bsy{\phi}_j$ are investigated in Section~\ref{subsec:oneatom}. These properties are then extended to $\|\bsy{\Phi}^{\mathrm{T}} \bsy{E} \bsy{Q} \|_{\infty} = \max_{j \in \left[n \right]} \left( \sum_{k=1}^{K} \left| \left\langle \bsy{\phi}_j, \bsy{e}_{k}  \right\rangle \right| q_k \right)$ in Section~\ref{subsec:natoms}. On the basis of this last result, Theorem~\ref{thm:finalGaussianThm} is obtained in Section~\ref{subsec:fullrecovery}. A particular case of the signal model is then examined in order to ease, in Section~\ref{sec:numresults}, the comparison of the results provided by means of the theoretical bound and those obtained by simulation.



\subsection{Assumptions and theoretical framework}

In this section, we restate some assumptions and define quantities used later on.\\

It is assumed that the entries of each $\bsy{e}_k$ ($1 \leq k \leq K$) are i.i.d. mean-zero Gaussian random variables of variance $\sigma_{k}^2$, \textit{i.e.}, $\bsy{e}_k \sim \mathcal{N}(0, \sigma_k^2 \bsy{I}_{m \times m})$. Furthermore, the noise vectors $\bsy{e}_k$ are statistically independent. Finally, the columns of $\bsy{\Phi}_S$ are assumed to be linearly independent. As already mentioned in Section~\ref{subsec:RIPDef}, the latter condition is true for all supports $S$ of cardinality $s$ whenever $\delta_{s} < 1$.\\

We define 
\begin{align}
& \bsy{q} := \begin{pmatrix}
q_1, & q_2, & \cdots, & q_K
\end{pmatrix}^{\mathrm{T}} \\
 & \bsy{\sigma} := \begin{pmatrix}
\sigma_1, & \sigma_2, & \cdots, & \sigma_K
\end{pmatrix}^{\mathrm{T}} \\
 & \bsy{q}^{(\bsy{\sigma})} := \begin{pmatrix}
\sigma_1 q_1, & \sigma_2 q_2, & \cdots, & \sigma_K q_K
\end{pmatrix}^{\mathrm{T}}.
\end{align}

Before going on further, we also define
\begin{equation}\label{eq:defKappa}
\kappa(\bsy{q}, \bsy{\sigma}) = \dfrac{1}{2 \| \bsy{q}^{(\bsy{\sigma})} \|_2^2}
\end{equation}
and
\begin{equation}\label{eq:defb}
b(\bsy{q}, \bsy{\sigma}) = \sqrt{\dfrac{2}{\pi}} \| \bsy{q}^{(\bsy{\sigma})} \|_1.
\end{equation}

Corollary \ref{corr:RIPBasedBtBound} shows that a sufficient condition for \mbox{SOMP-NS} to choose a correct atom at iteration $0$ is given by

\begin{equation}
\begin{aligned}
\max_{j \in \left[n \right]} \left( \sum_{k=1}^{K} \left| \left\langle \bsy{\phi}_j, \bsy{e}_{k}  \right\rangle \right| q_k \right) <   0.5 \left(1 - \|\bsy{\Phi}_S^+ \bsy{\Phi}_{\overline{S}} \|_{1} \right) 
  (1-\delta_{|S|})  \min_{j \in S} \sum_{k=1}^{K} | X_{j,k} | q_k.
\end{aligned}
\end{equation}

Using the upper bound $\left| \left\langle \bsy{\phi}_j, \bsy{e}_{k}  \right\rangle \right| \leq \| \bsy{\phi}_j \|_2 \| \bsy{e}_{k} \|_2 = \| \bsy{e}_{k} \|_2$ is a recurrent solution in the literature and provides satisfactory performance in a SMV setting (see \cite{cai2011orthogonal}). However, using such an approach in a MMV setting does not properly capture the performance gains obtained whenever $K$ increases as this upper bound assumes that all the error vectors $\bsy{e}_{k}$ are aligned with a single atom at each iteration. While this approximation is acceptable whenever only one measurement vector is available, it proves to be highly pessimistic as soon as one considers many independent (and usually isotropic) error vectors.\\

This observation motivates an in-depth analysis of the statistical properties of \newline $\max_{j \in \left[n \right]} \left( \sum_{k=1}^{K} \left| \left\langle \bsy{\phi}_j, \bsy{e}_{k}  \right\rangle \right| q_k \right)$. The analysis conducted hereafter mainly relies on the notion of Lipschitz functions and the related concentration inequalities.

\subsection{Concentration inequalities for one atom}\label{subsec:oneatom}

In this section, we are interested in providing a lower bound for $\mathbb{P} \left(\sum_{k=1}^{K} \left| \left\langle \bsy{\phi}_j, \bsy{e}_{k}  \right\rangle \right| q_k \geq \varepsilon \right)$ for an arbitrary atom $\bsy{\phi}_j$ \mbox{($1 \leq j \leq n$)}. The main result of this section is Lemma~\ref{lem:finalBoundIndivProb}.

\begin{thm}\label{thm:concentrationIneqLipschitz}
\cite[Theorem 8.40.]{foucart2013mathematical}
Let $f: \mathbb{R}^n \rightarrow \mathbb{R}$ be a Lipschitz function (with regards to the metric $\ell_2$) with Lipschitz constant $L$. Let $\bsy{g} = \big(
g_1,  g_2,  \dots,  g_n
\big)^{\mathrm{T}}$ be a vector of independent standard Gaussian random variables. Then, for all $\varepsilon > 0$
\begin{equation}
\mathbb{P} \left(f(\bsy{g}) - \mathbb{E}\lbrack f(\bsy{g})\rbrack \geq \varepsilon \right) \leq \exp \left( \dfrac{-\varepsilon^2}{2 L^2}\right),
\end{equation}
and consequently
\begin{equation*}
\mathbb{P} \left(\left|f(\bsy{g}) - \mathbb{E}\lbrack f(\bsy{g})\rbrack\right| \geq \varepsilon \right) \leq 2 \exp \left( \dfrac{-\varepsilon^2}{2 L^2}\right).
\end{equation*}
\end{thm}
The theorem above shows that Lipschitz functions $f(\bsy{g})$ tend to concentrate around their expectations when $\bsy{g}$ is distributed as a standard Gaussian random vector. Moreover, the concentration gets better as the Lipschitz constant $L$ decreases.\\

This theorem is intended to be used in conjunction with the function
\begin{equation}\label{eq:LipschitzFunc1}
f : \mathbb{R}^K \rightarrow \mathbb{R} : \bsy{g} \mapsto f(\bsy{g}) = \sum_{k=1}^K q_k \sigma_k |g_k|
\end{equation}
where $\sigma_k > 0$ ($1 \leq k \leq K$). This function will be shown to be equivalent to $\sum_{k=1}^{K} \left| \left\langle \bsy{\phi}_j, \bsy{e}_{k}  \right\rangle \right| q_k$ when $\bsy{e}_k$ is Gaussian.\\

We now wish to establish that $f$ is a Lipschitz function, compute the associated Lipschitz constant and determine its expectation. Let $\bsy{x}, \bsy{y} \in \mathbb{R}^K$, then, using the reverse triangle inequality and the Cauchy-Schwarz inequality,
\begin{align*}
\left| f(\bsy{x}) - f(\bsy{y}) \right| & = \left| \sum_{k=1}^K q_k \sigma_k |x_k| - \sum_{k=1}^K q_k \sigma_k |y_k| \right| \\
 & \leq \sum_{k=1}^K q_k \sigma_k \left| x_k - y_k \right| \\
 & \leq \| \bsy{q}^{(\bsy{\sigma})} \|_2 \| \bsy{x} - \bsy{y} \|_2
\end{align*}
Therefore, a valid Lipschitz constant $L$ of $f$ is equal to $\| \bsy{q}^{(\bsy{\sigma})} \|_2$. This is the best Lipschitz constant since, for $\bsy{y} = \bsy{0}$ and $\bsy{x} = \bsy{q}^{(\bsy{\sigma})}$, $| f(\bsy{x}) - f(\bsy{y}) | = \| \bsy{q}^{(\bsy{\sigma})} \|_2 \| \bsy{x} - \bsy{y} \|_2$.\\

Using the concentration inequalities for Lipschitz functions requires to know the value of $\mathbb{E}\left[f(\bsy{g}) \right]$. Using the linearity of the expectation and the fact that for $g \sim \mathcal{N}(0,1)$, $\mathbb{E} \lbrack |g| \rbrack = \sqrt{2/\pi}$, one easily obtains
\begin{equation*}
\mathbb{E}\left[f(\bsy{g}) \right] = \sqrt{\dfrac{2}{\pi}} \| \bsy{q}^{(\bsy{\sigma})}\|_1 = b(\bsy{q}, \bsy{\sigma}).
\end{equation*}
By using Theorem~\ref{thm:concentrationIneqLipschitz}, it is now possible to conclude that, for $\varepsilon > 0$,
\begin{equation}
\mathbb{P} \left(f(\bsy{g}) - b(\bsy{q}, \bsy{\sigma}) \geq \varepsilon \right) \leq \exp \left( - \kappa(\bsy{q}, \bsy{\sigma}) \varepsilon^2 \right).
\end{equation}
The final step of the development is to prove that the function $f$ defined above is distributed as $f(\bsy{g})$ for $\bsy{g} \sim \mathcal{N}(0, \bsy{I}_{m \times m})$. Indeed, $\sigma_k g_k \sim \mathcal{N}(0, \sigma_k^2)$ and $ \left\langle \bsy{\phi}_j, \bsy{e}_{k} \right\rangle = \sum_{i=1}^m (\bsy{\phi}_j)_i (\bsy{e}_{k})_i \sim \mathcal{N}(0,  \|\bsy{\phi}_j\|_2^2 \sigma_k^2) \sim \mathcal{N}(0,  \sigma_k^2)$.

The following lemma summarizes the discussion above and will be used to establish the theorem of the next section.

\begin{lem} \label{lem:finalBoundIndivProb}
Let $\bsy{\phi}_j \in \mathbb{R}^m$ where $\| \bsy{\phi}_j \|_2 = 1$ ($1 \leq j \leq n$). Let $\bsy{e}_k$ ($1 \leq k \leq K$) be independent random variables respectively distributed as $\mathcal{N}(0, \sigma_k^2 \bsy{I}_{m \times m})$ where $\sigma_k > 0$. It is also assumed that $q_k \geq 0$. Then, for $\varepsilon > 0$,
\begin{equation}
\mathbb{P} \left(\sum_{k=1}^{K} \left| \left\langle \bsy{\phi}_j, \bsy{e}_{k}  \right\rangle \right| q_k \geq b(\bsy{q}, \bsy{\sigma}) + \varepsilon \right) \leq \exp \left( - \kappa(\bsy{q}, \bsy{\sigma}) \varepsilon^2 \right).
\end{equation}
\end{lem}

It is important to highlight that Lemma \ref{lem:finalBoundIndivProb} provides an upper bound of the probability that $\sum_{k=1}^{K} \left| \left\langle \bsy{\phi}_j, \bsy{e}_{k}  \right\rangle \right| q_k$ is higher than $\varepsilon'  = b(\bsy{q}, \bsy{\sigma}) + \varepsilon$ only if $\varepsilon'$ is higher than $b(\bsy{q}, \bsy{\sigma})$.

\subsection{Concentration inequalities for $n$ atoms} \label{subsec:natoms}

The next problem to be tackled is that we would like to estimate an upper bound of the probability that $\max_{j \in \left[n \right]} ( \sum_{k=1}^{K} | \langle \bsy{\phi}_j, \bsy{e}_{k}  \rangle | q_k )$ is higher than a threshold $\varepsilon'$ instead of the same probability obtained for $\sum_{k=1}^{K} | \langle \bsy{\phi}_j, \bsy{e}_{k}  \rangle | q_k$. The issue we are facing is that the $n$ random variables $\sum_{k=1}^{K} | \langle \bsy{\phi}_j, \bsy{e}_{k}  \rangle | q_k$ (for $1 \leq j \leq n$) are not statistically independent which implies that the probability of upper bounding all these $n$ random variables simultaneously is not equal to the product of the probability to upper bound them separately. However, it remains possible to find a workaround using the union bound, as demonstrated by the following theorem.

\begin{thm}\label{thm:finalBoundManyAtomsProb}
Let $\bsy{\phi}_j \in \mathbb{R}^m$ where $\| \bsy{\phi}_j \|_2 = 1$ ($1 \leq j \leq n$). Let $\bsy{e}_k$ ($1 \leq k \leq K$) be independent random variables respectively distributed as $\mathcal{N}(0, \sigma_k^2 \bsy{I}_{m \times m})$ where $\sigma_k > 0$. It is also assumed that $q_k \geq 0$. Then, for $\varepsilon > 0$,
\begin{equation}
\mathbb{P} \left[\|\bsy{\Phi}^{\mathrm{T}} \bsy{E}^{(t)} \bsy{Q} \|_{\infty} \geq b(\bsy{q}, \bsy{\sigma}) + \varepsilon \right] \leq n \exp \left( - \kappa(\bsy{q}, \bsy{\sigma})  \varepsilon^2 \right).
\end{equation}
Equivalently,
\begin{equation}
\mathbb{P} \left[\|\bsy{\Phi}^{\mathrm{T}} \bsy{E}^{(t)} \bsy{Q} \|_{\infty} \leq b(\bsy{q}, \bsy{\sigma}) + \varepsilon \right] \geq 1 - n \exp \left( - \kappa(\bsy{q}, \bsy{\sigma}) \varepsilon^2 \right).
\end{equation}
\end{thm}
\begin{proof}
First of all, we observe that $\|\bsy{\Phi}^{\mathrm{T}} \bsy{E}^{(t)} \bsy{Q} \|_{\infty} = \max_{j \in \left[n \right]} \left( \sum_{k=1}^{K} \left| \left\langle \bsy{\phi}_j, \bsy{e}_{k}  \right\rangle \right| q_k \right)$.
Then, by union bound,
\begin{align*}
 & \mathbb{P} \left[\max_{j \in \left[n \right]} \left( \sum_{k=1}^{K} \left| \left\langle \bsy{\phi}_j, \bsy{e}_{k}  \right\rangle \right| q_k \right) \geq b(\bsy{q}, \bsy{\sigma}) + \varepsilon \right]  \\
 = & \mathbb{P} \left[ \bigcup_{j=1}^{n} \left[  \left( \sum_{k=1}^{K} \left| \left\langle \bsy{\phi}_j, \bsy{e}_{k}  \right\rangle \right| q_k \right) \geq b(\bsy{q}, \bsy{\sigma}) + \varepsilon \right] \right] \\
 \leq & \sum_{j=1}^{n} \mathbb{P} \left[ \left( \sum_{k=1}^{K} \left| \left\langle \bsy{\phi}_j, \bsy{e}_{k}  \right\rangle \right| q_k \right) \geq b(\bsy{q}, \bsy{\sigma}) + \varepsilon \right] \\
  \leq & n \exp \left( - \kappa(\bsy{q}, \bsy{\sigma}) \varepsilon^2 \right).
\end{align*}
The first inequality results from the union bound while the second inequality holds because of Lemma \ref{lem:finalBoundIndivProb}.
\end{proof}

\subsection{Full recovery of the support}\label{subsec:fullrecovery}
Theorem~\ref{thm:finalBoundManyAtomsProb} implicitly provides a lower bound on the probability of making correct decisions during iteration $0$. It is indeed possible to use either Equation~(\ref{eq:EqRIPBound}) or (\ref{eq:EqMuBound}) in conjunction with Equation~(\ref{eq:ProbInequalities}) and Theorem~\ref{thm:finalBoundManyAtomsProb} to obtain the desired lower bound. Nevertheless, a lower bound on the probability that SOMP-NS makes correct decisions from iteration $0$ to iteration $s < |S|$ remains to be found. This is the purpose of this section.\\

Before establishing the main result, one needs Lemmas~\ref{lem:fullSupportLem1} and \ref{lem:fullSupportLem2}. \\

\begin{lem}[On the distribution of $\langle \bsy{\phi}_j, (\bsy{I} - \bsy{P}) \bsy{e}_k \rangle$]\label{lem:fullSupportLem1}
If $\bsy{e}_{k}$ is distributed as $\mathcal{N}(0, \sigma_k^2 \bsy{I}_{m \times m})$ and $\bsy{P}$ is a fixed orthogonal projector matrix, then
$\left\langle \bsy{\phi}_j, (\bsy{I} - \bsy{P}) \bsy{e}_k \right\rangle$ is distributed as $\mathcal{N}(0, \tilde{\sigma}_k^2)$ where $\tilde{\sigma}_k \leq \sigma_k \|\bsy{\phi}_j \|_2 = \sigma_k$.
\end{lem}
\begin{proof} We have
\begin{equation*}
\left\langle \bsy{\phi}_j, (\bsy{I} - \bsy{P})\bsy{e}_{k} \right\rangle = \left\langle (\bsy{I} - \bsy{P}) \bsy{\phi}_j, \bsy{e}_{k} \right\rangle.
\end{equation*}
An auxiliary atom $\tilde{\bsy{\phi}}_j$ can be defined for index $j$, $\tilde{\bsy{\phi}}_j := (\bsy{I} - \bsy{P}) \bsy{\phi}_j$. This atom satisfies the inequality $\| \tilde{\bsy{\phi}}_j \|_2 \leq \| \bsy{\phi}_j \|_2 = 1$. Let us now observe that if the entries of $\bsy{e}_k$ are i.i.d. random variables distributed as $\mathcal{N}(0, \sigma_k^2)$, then $\langle \tilde{\bsy{\phi}}_j, \bsy{e}_{k} \rangle \sim  \mathcal{N}(0,  \|\tilde{\bsy{\phi}}_j\|_2^2 \sigma_k^2))$. The result immediately follows.
\end{proof}

Lemma $\ref{lem:fullSupportLem1}$ basically establishes that, for a fixed projection matrix $\bsy{P}$, $\left\langle \bsy{\phi}_j, (\bsy{I} - \bsy{P}) \bsy{e}_k \right\rangle$ is distributed as a mean-zero Gaussian random variable whose variance is always lower than that of the entries of $\bsy{e}_k $. This result will enable us to apply a statistical analysis similar to that of the previous section by keeping the Gaussian hypothesis.

\begin{lem}[Upper bounding the sum of random variables]\label{lem:fullSupportLem2}
We consider the random variables $X$, $Y_1 \sim \mathcal{N}(0, \sigma_{Y_1}^2)$ and $Y_2 \sim \mathcal{N}(0, \sigma_{Y_2}^2)$ where $\sigma_{Y_1} > \sigma_{Y_2} \geq 0$. It is assumed that $X \independent Y_1$ and $X \independent Y_2$. Then, for all $\varepsilon \in \mathbb{R}$, 
\begin{equation}
\mathbb{P} \left[ X + |Y_1| \leq \varepsilon  \right] \leq \mathbb{P} \left[ X + |Y_2| \leq \varepsilon  \right].
\end{equation}
\end{lem}

\begin{proof}
We know that $F_{|Y_1|}(y) = \mathrm{erf} (y/(\sqrt{2} \sigma_{Y_1}))$ and $F_{|Y_2|}(y) = \mathrm{erf} (y/(\sqrt{2} \sigma_{Y_2}))$ where 
\begin{equation*}
\mathrm{erf} \left( x \right) := \dfrac{2}{\sqrt{\pi}} \int_{0}^{x} \exp \left( - \xi^2 \right) \mathrm{d}\xi.
\end{equation*}

Since $\mathrm{erf}(x)$ is a monotonically increasing function, one obtains $\forall y > 0, \, F_{|Y_1|}(y) < F_{|Y_2|}(y)$. Thus,
\begin{align*}
\mathbb{P} \left[X + |Y_1| \leq \varepsilon  \right] & = \int_{-\infty}^\varepsilon \mathbb{P} \big[ |Y_1| \leq \varepsilon - x \big] f_X(x) \mathrm{d}x \\
 & \leq \int_{-\infty}^\varepsilon \mathbb{P} \big[ |Y_2| \leq \varepsilon - x \big] f_X(x) \mathrm{d}x\\
 & = \mathbb{P} \left[X + |Y_2| \leq \varepsilon  \right]. \qedhere
\end{align*}
\end{proof}

In the rest of this paper, the random variable $X$ of Lemma~\ref{lem:fullSupportLem2} will be replaced with a sum of $K-1$ independent random variables exhibiting half-normal distributions, \textit{i.e.}, $X = \sum_{k=1}^{K-1} |X_k|$ where the $X_k$ ($1 \leq k \leq K-1$) are independent and $X_k \sim \mathcal{N}(0, \sigma_k^2)$. Thereby, an immediate corollary of the lemma above is that the probability of upper bounding by $\varepsilon > 0$ the sum of statistically independent random variables exhibiting half-normal distributions is always decreased whenever the variance of at least one of the random variables is increased.\\

Let us now state the key theoretical result of this paper, which provides a lower bound on the probability that SOMP-NS selects correct atoms during the first $s+1$ iterations.

\begin{thm}[A lower bound on the probability of partial or full support recovery by SOMP-NS]\label{thm:finalGaussianThm}
Let $\bsy{e}_k$ ($1 \leq k \leq K$) be statistically independent Gaussian random vectors respectively distributed as $\mathcal{N}(0, \sigma_k^2 \bsy{I}_{m \times m})$. Let $\bsy{E} = \big(\bsy{e}_1, \dots , \bsy{e}_K \big)$. Let $\bsy{X} \in \mathbb{R}^{m \times K}$ be the signal matrix whose support is $S$. Let also $\bsy{\phi}_1, \dots, \bsy{\phi}_n$ be unit-norm vectors in $\mathbb{R}^m$ and $\bsy{\Phi} = \big( \bsy{\phi}_1, \dots, \bsy{\phi}_n \big)$ be the corresponding matrix which is assumed to satisfy the RIP with $|S|$-th RIC $\delta_{|S|} < 1$. Let $q_1, \dots, q_K \geq 0$. Then, SOMP-NS with dictionary matrix $\bsy{\Phi}$, weights $q_k$ ($1 \leq k \leq K$) and signal $\bsy{Y} = \bsy{\Phi} \bsy{X} + \bsy{E}$ is ensured to make correct decisions during the first $s+1$ iterations, \textit{i.e.}, from iteration $0$ to iteration $s$ included, with probability higher than
\begin{equation}\label{eq:finalTheoremProb}
1 - n \mathcal{C}_s \exp \left( - \kappa(\bsy{q}, \bsy{\sigma}) \overline{\varepsilon}(\bsy{q}, \bsy{\sigma})^2 \right)
\end{equation}
whenever $\overline{\varepsilon}(\bsy{q}, \bsy{\sigma}) := \varepsilon(\bsy{q}) - b(\bsy{q}, \bsy{\sigma}) > 0$ where
\begin{gather}
\mathcal{C}_s = \sum_{i=0}^{s} \binom{|S|}{i}, \\
\varepsilon(\bsy{q}) := 0.5 \left(1 - \|\bsy{\Phi}_S^+ \bsy{\Phi}_{\overline{S}}    \|_{1} \right) (1-\delta_{|S|})  \min_{j \in S} \sum_{k=1}^{K} | X_{j,k} | q_k,
\end{gather}
$\kappa(\bsy{q}, \bsy{\sigma})$ and $b(\bsy{q}, \bsy{\sigma})$ being defined in Equation~(\ref{eq:defKappa}) and (\ref{eq:defb}), respectively.
\end{thm}

\begin{proof}
First, we observe that $\langle \bsy{\phi}_j, (\bsy{I} - \bsy{P}^{(t)} )\bsy{e}_{k}  \rangle = \langle (\bsy{I} - \bsy{P}^{(t)} ) \bsy{\phi}_j, \bsy{e}_{k}  \rangle$ where it is convenient to define the $j$-th auxiliary atom at iteration $t$ as $\bsy{\phi}_j^{(t)} = (\bsy{I} - \bsy{P}^{(t)} ) \bsy{\phi}_j$. According to (\ref{eq:ProbInequalities}), at iteration $t$ (for $0 \leq t \leq s$), a sufficient condition for \mbox{SOMP-NS} to make a correct decision at iteration $t$ is given by 
\begin{equation*}
\begin{aligned}
\max_{j \in \left[n \right]} \left( \sum_{k=1}^{K} \left| \left\langle \bsy{\phi}_j^{(t)}, \bsy{e}_{k}  \right\rangle \right| q_k \right) \leq  0.5 (1 - \|\bsy{\Phi}_S^+ \bsy{\Phi}_{\overline{S}} \|_{1} ) 
 (1-\delta_{|\mathcal{J}_t|}) \min_{j \in S} \sum_{k=1}^{K} | X_{j,k} | q_k.
\end{aligned}
\end{equation*}
Since $a \leq b$ implies $\delta_a \leq \delta_b$, then $1-\delta_{|\mathcal{J}_t|} \geq 1-\delta_{|S|}$ which shows that a less tight sufficient condition for SOMP-NS to make a correct decision at iteration $t$ is
\begin{equation}\label{eq:suffCondWSOMPStept}
\begin{aligned}
\max_{j \in \left[n \right]} \left( \sum_{k=1}^{K} \left| \left\langle \bsy{\phi}_j^{(t)}, \bsy{e}_{k}  \right\rangle \right| q_k \right) \leq \varepsilon (\bsy{q}) = 0.5 (1 - \|\bsy{\Phi}_S^+ \bsy{\Phi}_{\overline{S}} \|_{1} )
 (1-\delta_{|S|})  \min_{j \in S} \sum_{k=1}^{K} | X_{j,k} | q_k.
\end{aligned}
\end{equation}
Therefore, if the condition above holds for $t=0$, SOMP-NS is guaranteed to pick a correct atom at iteration $0$. At iteration $1$, the orthogonal projector $\bsy{P}^{(1)}$ can take $|S|$ different values (each value of $\bsy{P}^{(1)}$ corresponds to a specific support $S_1 \subset S$). Therefore, if (\ref{eq:suffCondWSOMPStept}) holds for every possible projector matrix at iteration $1$, SOMP-NS is guaranteed to pick a correct atom at iteration $1$. Thus, we need to satisfy $|S| n + 1$ equations of the form $\sum_{k=1}^{K} \left| \left\langle \bsy{\phi}, \bsy{e}_{k}  \right\rangle \right| q_k \leq \varepsilon(\bsy{q})$ (where $\| \bsy{\phi} \|_2 \leq 1$) to ensure that SOMP-NS picks correct atoms during the first two iterations.\\

Extending the previous train of thought from iteration $0$ to iteration $s$, one easily comes to the conclusion that $\mathcal{C}_{s} n$ equations of the form $\sum_{k=1}^{K} \left| \left\langle \bsy{\phi}, \bsy{e}_{k}  \right\rangle \right| q_k \leq \varepsilon(\bsy{q})$ (where $\| \bsy{\phi} \|_2 \leq 1$) should be satisfied since, at iteration $t$, there exist $\binom{|S|}{t}$ possible realizations of the orthogonal projector $\bsy{P}^{(t)}$. Using the union bound as in the proof of Theorem~\ref{thm:finalBoundManyAtomsProb} shows that the probability of satisfying the $\mathcal{C}_s n$ equations is lower bounded by
\begin{equation}\label{eq:trueEquationModel}
1 - \sum_{t=0}^{s} \sum_{i=1}^{\binom{|S|}{t}} \sum_{j=1}^{n} \mathbb{P} \left[ \sum_{k=1}^{K} \left| \left\langle (\bsy{I} - \bsy{P}^{(t,i)}) \bsy{\phi}_j, \bsy{e}_{k}  \right\rangle \right| q_k \geq \varepsilon(\bsy{q}) \right]
\end{equation}
where $\bsy{P}^{(t,i)}$ is the $i$-th possible realization of the orthogonal projector at iteration $t$ assuming that only correct atoms have been picked before iteration $t$.\\

Lemmas \ref{lem:fullSupportLem1} and \ref{lem:fullSupportLem2} imply that
\begin{equation}
\begin{aligned}
\mathbb{P} \left[ \sum_{k=1}^{K} \left| \left\langle (\bsy{I} - \bsy{P}^{(t,i)})\bsy{\bsy{\phi}_j}, \bsy{e}_{k}  \right\rangle \right| q_k \geq \varepsilon(\bsy{q}) \right] \leq 
 \mathbb{P} \left[ \sum_{k=1}^{K} \left| \left\langle \bsy{\bsy{\phi}_j}, \bsy{e}_{k}  \right\rangle \right| q_k \geq \varepsilon(\bsy{q}) \right].
\end{aligned}
\end{equation}
Furthermore, since $\varepsilon(\bsy{q}) = \overline{\varepsilon}(\bsy{q}, \bsy{\sigma}) + b(\bsy{q}, \bsy{\sigma})$, Lemma \ref{lem:finalBoundIndivProb} shows that
\begin{equation*}
\mathbb{P} \left[\sum_{k=1}^{K} \left| \left\langle \bsy{\phi}_j, \bsy{e}_{k}  \right\rangle \right| q_k \geq \varepsilon(\bsy{q}) \right] \leq \exp \left( - \kappa(\bsy{q}, \bsy{\sigma})  \overline{\varepsilon}(\bsy{q}, \bsy{\sigma})^2 \right).
\end{equation*}
whenever $\overline{\varepsilon}(\bsy{q}, \bsy{\sigma}) > 0$.
Finally, the probability that SOMP-NS chooses correct atoms from iteration $0$ to iteration $s$ (let us denote this event by $E_{\mathrm{succ}}^{(s)}$) satisfies
\begin{equation*}
\mathbb{P} \left[E_{\mathrm{succ}}^{(s)} \right] \geq 1 - n \mathcal{C}_s \exp \left( - \kappa(\bsy{q}, \bsy{\sigma}) \overline{\varepsilon}(\bsy{q}, \bsy{\sigma})^2 \right). \qedhere
\end{equation*}
\end{proof}

The theorem above translates several intuitive realities. First, by examining the expression of $\varepsilon$, one concludes that several quantities influence the probability of correct recovery:
\begin{itemize}
\item The expression $\|\bsy{\Phi}_S^+ \bsy{\Phi}_{\overline{S}}    \|_{1}$ quantifies the robustness of the recovery in the noiseless case. It is evident that the margin of error when deciding which atom to choose in the noiseless case will contribute to determine the admissible noise level in the noisy case.
\item The term $(1-\delta_{|S|})$ translates the ability of the dictionary matrix $\bsy{\Phi}$ to maintain the $\ell_2$ norm of $|S|$-sparse signals, which ensures that the norm of the columns of $\bsy{Z}^{(t)}$ remain sufficiently high to avoid being absorbed by the noise matrix $\bsy{E}^{(t)}$.
\item The minimized sum $\min_{j \in S} \sum_{k=1}^{K} | X_{j,k} | q_k$ depicts the idea that the weighted sum of the coefficients associated with every atom should be high enough to allow SOMP-NS to identify them when noise is added to the measurements. This term simultaneously  captures the influence of the signal to be recovered as well as that of the weights.
\end{itemize}


Moreover, $\kappa(\bsy{q}, \bsy{\sigma})$ indicates that increasing the noise variances decreases the probability of support recovery. Also, increasing the weights naturally augments the power of the noise signal $\bsy{E}^{(t)}$. Nevertheless, it should be expected that this effect is counterbalanced by $\varepsilon(\bsy{q})$ as the latter variable is also a function of the weights. Finally, $n$ translates the fact that the noise affects every atom while $\mathcal{C}_s$ takes into account the existence of $s$ iterations.\\

It is worth noticing that $\mathcal{C}_{|S|-1} = 2^{|S|}-1 \leq 2^{|S|}$ so that it indicates that the probability of full support recovery is lower bounded by 
\begin{equation*}
1 - n 2^{|S|} \exp \left( - \kappa(\bsy{q}, \bsy{\sigma}) \overline{\varepsilon}(\bsy{q}, \bsy{\sigma})^2 \right).
\end{equation*}

On the basis of past results in the literature, it will be suggested in Section~\ref{subsec:GribonvalrelatedThm} that $\mathcal{C}_s$ is an artifact of our proof and should ideally be replaced with a function that increases more slowly as $|S|$ is augmented.\\

Also, as a last remark, Theorem~\ref{thm:finalGaussianThm} can be rephrased by stating that the probability of failure of SOMP-NS from iteration $0$ to iteration $s$, \textit{i.e.}, at least one wrong atom is chosen during the first $s+1$ iterations, admits the upper bound
\begin{equation}
n \mathcal{C}_s \exp \left( - \kappa(\bsy{q}, \bsy{\sigma}) \overline{\varepsilon}(\bsy{q}, \bsy{\sigma})^2 \right).
\end{equation}

Theorem~\ref{thm:finalGaussianThm} is now to be particularized using the coherence of $\bsy{\Phi}$ instead of the RIP and the ERC.

\begin{thm}[A coherence-based lower bound on the probability of full support recovery by SOMP-NS]\label{thm:finalGaussianThmCoherence}
Let $\bsy{e}_k$ ($1 \leq k \leq K$) be statistically independent Gaussian random vectors respectively distributed as $\mathcal{N}(0, \sigma_k^2 \bsy{I}_{m \times m})$. Let $\bsy{E} = \big(
\bsy{e}_1, \dots, \bsy{e}_K \big)$. Let also $\bsy{X} \in \mathbb{R}^{m \times K}$ be the signal matrix whose support is $S$. Let also $\bsy{\phi}_1, \dots, \bsy{\phi}_n$ be unit-norm vectors in $\mathbb{R}^m$ and $\bsy{\Phi} = \big( \bsy{\phi}_1, \dots, \bsy{\phi}_n \big)$ be the corresponding matrix whose coherence $\mu$ is assumed to satisfy $\mu < \frac{1}{2 |S| -1}$. Let $q_1, \dots, q_K \geq 0$. Then, SOMP-NS with dictionary matrix $\bsy{\Phi}$, weights $q_k$ ($1 \leq k \leq K$) and signal $\bsy{Y} = \bsy{\Phi} \bsy{X} + \bsy{E}$ is ensured to make correct decisions during the first $s+1$ iterations, \textit{i.e.}, from iteration $0$ to iteration $s$ included, with probability higher than
\begin{equation}
1 - n \mathcal{C}_s \exp \left( - \kappa(\bsy{q}, \bsy{\sigma}) \overline{\varepsilon}^{(\mu)}(\bsy{q}, \bsy{\sigma})^2 \right)
\end{equation}
whenever $\overline{\varepsilon}^{(\mu)}(\bsy{q}, \bsy{\sigma}) := \varepsilon^{(\mu)}(\bsy{q}) - b(\bsy{q}, \bsy{\sigma}) > 0$, where
\begin{gather}
\mathcal{C}_s = \sum_{i=0}^{s} \binom{|S|}{i}, \nonumber \\
\varepsilon^{(\mu)}(\bsy{q}) := 0.5 (1 - \mu (2|S|-1))  \min_{j \in S} \sum_{k=1}^{K} | X_{j,k} | q_k,
\end{gather}
$\kappa(\bsy{q}, \bsy{\sigma})$ and $b(\bsy{q}, \bsy{\sigma})$ being defined in Equation~(\ref{eq:defKappa}) and (\ref{eq:defb}), respectively.
\end{thm}
\begin{proof}
As it has already been pointed out in Section~\ref{subsec:ERCNoiseless}, if $\mu < \frac{1}{2 |S| -1}$, then $ \|\bsy{\Phi}_S^+ \bsy{\Phi}_{\overline{S}} \|_{1} \leq \frac{|S| \mu}{1 - (|S|-1) \mu} < 1$ for all the supports of size $|S|$. Moreover, Equation~(\ref{eq:RICToCoherence}) shows that $\delta_s \leq (s-1) \mu$ whenever $\mu < 1/((s-1))$ Thus, using the two inequalities above in conjunction with Theorem~\ref{thm:finalGaussianThm} yields
\begin{equation*}
\begin{aligned}
\varepsilon^{(\mu)}(\bsy{q}) = 0.5 \left(1 - \frac{|S| \mu}{1 - (|S|-1) \mu} \right) (1-(|S|-1) \mu)
  \min_{j \in S} \sum_{k=1}^{K} | X_{j,k} | q_k.
\end{aligned}
\end{equation*}
Elementary algebraic manipulations show that the expression above simplifies to
\begin{equation*}
\varepsilon^{(\mu)}(\bsy{q}) = 0.5 (1 - \mu (2|S|-1))  \min_{j \in S} \sum_{k=1}^{K} | X_{j,k} | q_k. \qedhere
\end{equation*}
\end{proof}

A similar although stronger result can be obtained by means of the cumulative coherence function $\mu_1$. However, the resulting expression of $\varepsilon^{(\mu)}(\bsy{q})$ is more complicated. Moreover, coherence-based bounds only have a theoretical interest as they prove to be pessimistic for practical cases, even when using the cumulative coherence function.

\subsection{Related results}\label{subsec:GribonvalrelatedThm}

A result similar to Theorem~\ref{thm:finalGaussianThm} has already been obtained in \cite{gribonval2008atoms}. The striking similarities between our result and that obtained in \cite{gribonval2008atoms} motivate this section. They prove the following theorem.

\begin{thm} \cite[Theorem 7]{gribonval2008atoms} \label{thm:gribonval}
Let $\bsy{Y} = \bsy{\Phi}_S \Sigma^{1/2} \bsy{U} + \bsy{E} \in \mathbb{R}^{m \times K}$ with $\bsy{U}$ a $|S| \times K$ matrix of standard Gaussian random variables, $\Sigma = \mathrm{diag}(\sigma_i^2)_{i \in S}$ and $\bsy{E}$ an error term orthogonal to the atoms in $|S|$. Assume that the dictionary matrix $\bsy{\Phi}$ satisfies the restricted isometry property (RIP) with restricted isometry constant (RIC) (of order $|S|+1$) $\delta_{|S|+1} < 1/3$ and
\begin{equation*}
\|\bsy{\Phi}_{\overline{S}}^{\mathrm{T}} \bsy{E}\|_{\infty} < \sqrt{\dfrac{2}{\pi}} K \min_{i \in S} \sigma_i (1-3 \delta_{|S|+1}).
\end{equation*}
Then, the probability that $|S|$ iterations of SOMP fail to exactly recover the support $S$ on the basis of $\bsy{Y}$ does not exceed $n 2^{|S|} \exp \left( - K \gamma^2/ \pi \right)$ with $K$ the number of measurement vectors, $n$ the number of atoms and
\begin{equation*}
\gamma := 1 - 3 \delta_{|S|+1} - \left( \sqrt{\dfrac{2}{\pi}} K \min_{i \in S} \sigma_i \right)^{-1} \|\bsy{\Phi}_{\overline{S}}^{\mathrm{T}} \bsy{E}\|_{\infty}.
\end{equation*}
\end{thm}

In \cite{gribonval2008atoms}, the statistical analysis has been focused on the sparse signal to be estimated, \textit{i.e.}, $\bsy{X}$, while no particular statistical distribution is assumed for $\bsy{E}$. Since the noise vectors are assumed to be orthogonal to the columns of $\bsy{X}$, their purpose is to model an approximation noise, \textit{i.e.}, the part of the signal $\bsy{Y}$ that cannot be mapped by $\bsy{\Phi} \bsy{X}$. Conversely, the noise signals envisioned in our paper represent additive measurement noises that cannot be assumed to be orthogonal to any vector subspace.\\

In the noiseless case, a result similar to Theorem~\ref{thm:gribonval} has been obtained in \cite[Theorem 6.2]{eldar2010average} for a variant of SOMP entitled 2-SOMP, \textit{i.e.}, a SOMP algorithm where the decision on which atom to pick is performed on the basis of the maximization of a $\ell_2$ norm instead of a $\ell_1$ norm.\\

Although the problem addressed in this paper is different from that of Theorem~\ref{thm:gribonval}, the authors of \cite{gribonval2008atoms, eldar2010average} have used approaches similar to ours and they also noticed that the parasitic term $\mathcal{C}_s$ seems difficult to remove. The rationale of this discussion is thus that it is difficult to avoid the suboptimal term $\mathcal{C}_s$ when conducting developments similar to that presented in this paper.

\subsection{Conjectures}\label{subsec:conjectures}

Theorem~\ref{thm:finalGaussianThm} provides interesting insights into how successful SOMP-NS is whenever some parameters are modified. However, the theoretical developments presented in this paper do not properly capture all the characteristics of SOMP-NS.\\

First of all, regarding $n \mathcal{C}_{|S|-1}$, the value of $n$ should be lowered in practice. The reason why $n$ should be replaced by $\overline{n} \ll n$ is because our analysis assumes that all the atoms not to be picked (of index $j$) are such that $\sum_{k=1}^{K} | \langle \bsy{\phi}_j, \bsy{x}_{k}^{(t)}  \rangle | q_k  = \max_{z \in \overline{S}} ( \sum_{k=1}^{K} | \langle \bsy{\phi}_z, \bsy{x}_{k}^{(t)}  \rangle | q_k )$ while, in practice, only a few atoms are likely to exhibit a significant value of $\sum_{k=1}^{K} | \langle \bsy{\phi}_j, \bsy{x}_{k}^{(t)}  \rangle | q_k$.\\

We conjecture that $\mathcal{C}_{|S|-1}$ may be replaced by a linear function of $|S|$. The reason why we failed at obtaining such a result is probably linked to the proof of Theorem \ref{thm:finalGaussianThm}. All the possible supports at each iteration are considered and it is ensured that the sufficient condition for making a correct decision is satisfied for each support. This is very pessimistic as only one support out of all the numerous possibilities actually matters. As indicated in Section~\ref{subsec:GribonvalrelatedThm}, other researchers working on similar problems have stumbled upon this issue and no solution has been found so far to the best of the authors' knowledge. The simulation results presented in the end of this paper will however not address this problem.\\

Furthermore, as it will be explained in Section~\ref{subsec:summaryNumres}, the bias $b(\bsy{q}, \bsy{\sigma})$ should be removed in order to deliver results compatible with what has been observed in our numerical simulations. This term is most likely an artifact linked to Equation~(\ref{eq:firstBigApprox}) and Equation~(\ref{eq:secondBigApprox}). Equation~(\ref{eq:firstBigApprox}) basically assumes that $\langle \bsy{\phi}_j, \bsy{x}_k \rangle$ and $\langle \bsy{\phi}_j, \bsy{e}_k \rangle$ always have opposite signs for atoms $\bsy{\phi}_j$ whose indexes $j$ belong to the support $S$. Conversely, Equation~(\ref{eq:secondBigApprox}) assumes that $\langle \bsy{\phi}_j, \bsy{x}_k \rangle$ and $\langle \bsy{\phi}_j, \bsy{e}_k \rangle$ always have identical signs for atoms whose indexes do not belong to the support $S$. Thus, a statistical analysis directly performed on $\sum_{k=1}^{K} |\langle \bsy{\phi}_j, \bsy{x}_k \rangle + \langle \bsy{\phi}_j, \bsy{e}_k \rangle| q_k $ may prevent the bias term from appearing. \\

Therefore, considering all the conjectures above, a better bound may be given by
\begin{equation}\tag{B1}\label{eq:B1}
1 - \overline{n} \alpha s \exp \left( - \kappa(\bsy{q}, \bsy{\sigma}) \overline{\varepsilon}(\bsy{q}, \bsy{\sigma})^2 \right)
\end{equation}
where $n$ has been replaced with $\overline{n} \ll n$, $\alpha s$ has been substituted to $\mathcal{C}_s$ and $b(\bsy{q}, \bsy{\sigma})$ has been canceled out.\\

Finally, it is worth pointing out that Equation~(\ref{eq:B1}) will likely not deliver perfect results. Indeed, Equation~(\ref{eq:trueEquationModel}) suggests that the probability of sparse support recovery success is actually a sum of exponential functions, each function corresponding to a single atom, iteration and support. Moreover, possibly different values of $\overline{\varepsilon}(\bsy{q}, \bsy{\sigma})^2$ should be chosen for each exponential function. It is also probably suboptimal to make use of the union bound.

\section{Numerical results}\label{sec:numresults}
This section aims at demonstrating that:
\begin{enumerate}
\item SOMP-NS provides significant performance improvements when compared to SOMP provided that the noise vectors exhibit different variances and that the weights are properly chosen.
\item Depending on the signal model that is chosen, it is possible to accurately estimate the optimal weights by using simple closed-formed formulas derived from Equation~(\ref{eq:B1}).
\item Equation~(\ref{eq:B1}) properly predicts the performance improvements obtained when the number of measurement vectors increases.
\end{enumerate}
The purpose of the first point is to demonstrate numerically that the gains provided by SOMP-NS are significant. The last two points rather focus on the numerical validation of the theoretical analysis presented in this paper. With these goals in mind, a particular signal model will be chosen. In particular, this signal model will be sufficiently simple to allow the computation of the optimal weights when the model~(\ref{eq:B1}) is assumed to be correct.

\subsection{Particular signal model}\label{subsec:specialSigModel}

The objective of this section is to particularize Theorem~\ref{thm:finalGaussianThm} to models for which $X_{i, k} = \varepsilon_{i, k} \mu_X$ ($1 \leq k \leq K$, $i \in S$) and $X_{i, k} = 0$ ($1 \leq k \leq K$, $i \not\in S$) where the terms $\varepsilon_{i, k}$ denote Rademacher random variables, \textit{i.e.}, random variables that return either $1$ or $-1$ with probability $0.5$ for both outcomes.\\

Two different sign patterns will be distinguished. Sign pattern 1 refers to the case where the sign pattern is identical for all the sparse vectors $\bsy{x}_k$ to be recovered, \textit{i.e.}, $\varepsilon_{i, k} = \varepsilon_{i, 1}$ for all $k \in \lbrack K \rbrack$, $i \in S$ and $\varepsilon_{i_1, 1} \independent \varepsilon_{i_2, 1}$ whenever $i_1 \neq i_2$. Sign pattern 2 corresponds to the situation where the sign pattern is independent for each and within each $\bsy{x}_k$, \textit{i.e.}, $\varepsilon_{i_1, k_1} \independent \varepsilon_{i_2, k_2}$ whenever $i_1 \neq i_2$ and/or $k_1 \neq k_2$.\\

In both cases, it is worth mentioning that the absolute values of the entries of each $\bsy{x}_k$ are equal to $\mu_X$. Thus, $\min_{j \in S} \sum_{k=1}^{K} | X_{j,k} | q_k = \|\bsy{q}\|_1 \mu_X$ and, according to Theorem~\ref{thm:finalGaussianThm}, the probability of full support recovery is always higher than 
\begin{equation}
1 - n \mathcal{C}_{|S|-1} \exp \left( -  \|\bsy{q}\|_1^2 \mu_X^2 \kappa(\bsy{q}, \bsy{\sigma}) \left( \varepsilon' - \dfrac{b(\bsy{q}, \bsy{\sigma})}{\|\bsy{q}\|_1 \mu_X} \right)^2 \right)
\end{equation}
where $\varepsilon' := 0.5 \left(1 - \|\bsy{\Phi}_S^+ \bsy{\Phi}_{\overline{S}}    \|_{1} \right) (1-\delta_{|S|})$. Also, the bound above only holds if $\varepsilon' > b(\bsy{q}, \bsy{\sigma})/(\|\bsy{q}\|_1 \mu_X)$.\\

By using the conjecture about the elimination of the bias term $b(\bsy{q}, \bsy{\sigma})$ described in Section~\ref{subsec:conjectures}, one obtains Equation~(\ref{eq:modelConjecture}), which is reminiscent of Equation~(\ref{eq:B1}), except that (\ref{eq:modelConjecture}) does not include the conjectures linked to the term $n \mathcal{C}_{|S|-1}$. If our only objective is to derive the optimal weights according the theoretical model, only the argument of the exponential matters so that adjusting the term $n \mathcal{C}_{|S|-1}$ has no effect.

\begin{equation}\label{eq:modelConjecture}
1 - n \mathcal{C}_{|S|-1} \exp \left( -  \|\bsy{q}\|_1^2 \mu_X^2 \kappa(\bsy{q}, \bsy{\sigma}) \varepsilon'^2 \right)
\end{equation}

Let $\langle \bsy{x} \rangle$ denote the arithmetic mean of the entries of vector $\bsy{x}$. Then, $\| \bsy{q} \|_1^2 = K^2  \langle \bsy{q} \rangle^2$ and $\| \bsy{q}^{(\bsy{\sigma})} \|_2^2 = K \langle (q_k^2 \sigma_k^2)_{k \in \lbrack K \rbrack} \rangle$. Thus, the probability of full support recovery is always higher than
\begin{equation}\tag{B2}\label{eq:B2}
1 - n \mathcal{C}_{|S|-1} \exp \left( -  \dfrac{K \langle \bsy{q} \rangle^2 \mu_X^2}{2 \langle (q_k^2 \sigma_k^2)_{k \in \lbrack K \rbrack} \rangle} \varepsilon'^2 \right).
\end{equation}
As a particular case, if $q_k = 1$ ($ 1 \leq k \leq K$), the latter probability also rewrites
\begin{equation}
1 - n \mathcal{C}_{|S|-1} \exp \left( -  \dfrac{K \mu_X^2}{2 \langle (\sigma_k^2)_{k \in \left[ K \right]}  \rangle} \varepsilon'^2 \right).
\end{equation}
Let us now focus our attention on the weights that maximize Equation~(\ref{eq:B2}) or, equivalently, $\langle \bsy{q} \rangle^2/\langle (q_k^2 \sigma_k^2)_{k \in \lbrack K \rbrack} \rangle$. We can restrict our attention to the maximization of 
\begin{equation*}
\dfrac{\langle \bsy{q} \rangle}{\sqrt{\langle (q_k^2 \sigma_k^2)_{k \in \lbrack K \rbrack} \rangle}} = \dfrac{ \sum_{k=1}^K q_k }{\sqrt{\sum_{k=1}^K q_k^2 \sigma_k^2}}.
\end{equation*}
We define $\tilde{q}_k = q_k \sigma_k$ so that the expression above also reads
\begin{equation*}
\dfrac{ \sum_{k=1}^K \tilde{q}_k (1/ \sigma_k) }{\| \bsy{\tilde{q}} \|_2} =  \left\langle \dfrac{\bsy{\tilde{q}}}{\| \bsy{\tilde{q}} \|_2}, \left( \dfrac{1}{\sigma_k} \right)_{k \in \lbrack K \rbrack} \right\rangle.
\end{equation*}
The quantity $\bsy{\tilde{q}}/\| \bsy{\tilde{q}} \|_2$ represents the direction of $\bsy{\tilde{q}}$ so that we know that a global maximizer is obtained whenever $\bsy{\tilde{q}}$ and $ \left( 1/\sigma_k \right)_{k \in \lbrack K \rbrack}$ have the same direction. It means that a global maximizer is obtained if and only if $\tilde{q}_k = C (1/\sigma_k)$ where $C > 0$. This is equivalent to requiring $q_k = C (1/\sigma_k^2)$  since $\tilde{q}_k = q_k \sigma_k$. By choosing $C = 1$, one concludes that $q_k = 1/ \sigma_k^2$ provides an optimal weighting strategy according to (\ref{eq:B2}).

\subsection{Simulation frameworks}

Now that the optimal weighting strategy for the signal models envisioned in Section~\ref{subsec:specialSigModel} has been derived, it becomes possible to determine whether the bound (\ref{eq:B2}) properly predicts the impact of the weights on the performance that is achieved.\\

Our MATLAB simulation software is available in \cite{determe2015SOMPNSSimPack}. All the scripts needed to generate the figures presented in this section are also available. The reader should know that every simulation result exposed hereafter has been performed by using single precision floating point representations. The reason for such a choice is that single precision arithmetic is faster and thus preferred for algorithms intended to run on real-time platform such as SOMP-NS. For the same reason, the simulation results are obtained faster when using the single precision format.\\

It is assumed that $m = 250$ and $n = 1000$. The simulation framework consists of a fixed dictionary matrix $\bsy{\Phi} \in \mathbb{R}^{250 \times 1000}$ whose entries were generated on the basis of independent and identically distributed Gaussian random variables and then normalized in such a way that each column of the matrix exhibits a unit $\ell_2$ norm. This matrix is fixed for all the simulations and is available in \cite{determe2015SOMPNSSimPack}. \\

Two simulation frameworks have been envisioned to demonstrate the three points introduced at the very beginning of Section~\ref{sec:numresults}. The first framework consists of simulations for the case $K = 2$ and addresses the first two objectives while the last framework examines how the probability of successful support recovery evolves as $K$ increases.

\subsection{Simulation results for $K = 2$}\label{subsec:simresK2}

First of all, it is worth pointing out that the performance achieved by \mbox{SOMP-NS} is invariant if the weight vector $\bsy{q} := (q_1, q_2)^{\mathrm{T}}$ is multiplied by a positive constant. Therefore, only the angle $\theta_q := \arctan (q_2/q_1)$ is investigated. The weights are thus generated on the basis of the polar coordinate system $q_1 = \cos \theta_q$, $q_2 = \sin \theta_q$ where $0^{\circ} \leq \theta_q \leq 90^{\circ}$. In practice, the grid of weighting angles $\theta_q$ will consist of $21$ uniformly spaced angles from $5^{\circ}$ to $85^{\circ}$. \\

The noise standard deviation vector $\bsy{\sigma} := (\sigma_1, \sigma_2)^{\mathrm{T}}$ is generated on the basis of the polar coordinate system $\bsy{\sigma} = ( \cos \theta_{\sigma}, \sin \theta_{\sigma})^{\mathrm{T}}$ where $\theta_{\sigma}$ describes the orientation of the noise vector. The grid of values for $\theta_{\sigma}$ is composed of $21$ uniformly spaced angles ranging from $20^{\circ}$ to $70^{\circ}$. Extremely high or low angles have been avoided because they correspond to situations for which the noise concentrates essentially on one measurement vector. Therefore, appropriate weighting strategies would be able to cancel most of the noise and would lead to probabilities of correct support recovery that are too high to be reliably estimated on the basis of a limited number of Monte Carlo cases.\\

A total of six simulation configurations have been run. Each configuration corresponds to one of the two sign patterns described in Section~\ref{subsec:specialSigModel}, to a support size $|S|$ and to a value of $\mu_X$. Simulation cases have been generated for each value of $\theta_{\sigma}$ belonging to the grid defined beforehand. Once the support size is fixed, the actual support is randomly and uniformly chosen among all the possibilities. The support that is simulated is independent for each case.\\

Although $\| \bsy{\sigma} \|_2^2 = 1$, the input signal-to-noise ratio (SNR), referred to as $\mathrm{SNR}_{\mathrm{in}}$ and defined in Equation~(\ref{eq:defSNRinput}), is to be modified by means of the quantity $\mu_X$. \\

\begin{equation}\label{eq:defSNRinput}
\mathrm{SNR}_{\mathrm{in}} \; \mathrm{(dB)} = 20 \log_{10} \left( \dfrac{\| \bsy{Y} \|_{\mathrm{F}}}{\| \bsy{Y} - \bsy{\Phi} \bsy{X} \|_{\mathrm{F}}} \right)
\end{equation}

Table \ref{tab:configsSimK2} describes all the configurations that have been investigated numerically. The values of $\mu_X$ have been chosen in such a way that the probability of full support recovery when $(\theta_q, \theta_{\sigma}) = (45^{\circ}, 45^{\circ})$ is approximately equal to $0.1$ and $0.3$ for the sign pattern 1 and sign pattern 2 respectively. These probabilities have been chosen in such a way that the probability of successful full support recovery over the grid defined for $(\theta_q, \theta_{\sigma})$ never reaches values so high that it cannot be reliably estimated on the basis of a limited number of simulation experiments. The lower value of the probability of successful recovery for sign pattern $1$ is linked to the fact that having a sign pattern independent for each measurement vector, as for sign pattern 2, provides performance improvements. This observation is actually reminiscent to what has been established in Theorem~\ref{thm:gribonval}.\\

In Table \ref{tab:configsSimK2}, the input SNR, \textit{i.e.}, $\mathrm{SNR}_{\mathrm{in}}$, is estimated by generating $2 \; 10^5$ cases for the configuration of interest, then the input SNR (in dB) is computed for each case and the results are finally averaged. The Matlab script implementing this estimation is available in \cite{determe2015SOMPNSSimPack}.

\begin{table}[!h]
\caption{Simulation configurations | $K = 2$}
\label{tab:configsSimK2}
\centering
\begin{tabular}{|p{3.6cm}|p{2.6cm}|p{0.5cm}|p{0.7cm}|p{1.5cm}|p{1.6cm}|}
\hline
\bfseries Configuration ID & \bfseries Sign pattern & \bfseries $|S|$ & \bfseries $\mu_X$ & \bfseries $\mathrm{SNR}_{\mathrm{in}}$ & \bfseries \# cases \\
\hline\hline
Configuration $1$ & $1$ & $10$ & $2.28$ & $1.51$ dB & $1.0 \; 10^4$\\
\hline
Configuration $2$ & $1$ & $30$ & $3.19$ & $5.37$ dB & $1.0 \; 10^4$\\
\hline
Configuration $3$ & $1$ & $40$ & $6.94$ & $12.15$ dB & $2.5 \; 10^4$\\
\hline\hline
Configuration $4$ & $2$ & $10$ & $2.50$ & $1.76$ dB & $1.0 \; 10^4$\\
\hline
Configuration $5$ & $2$ & $30$ & $3.06$ & $5.12$ dB & $1.0 \; 10^4$\\
\hline
Configuration $6$ & $2$ & $40$ & $3.42$ & $6.76$ dB & $1.0 \; 10^4$\\
\hline
\end{tabular}
\end{table}

As an example, Figure~\ref{fig:simAnglesPlotK2Conf2} depicts the probability of full support recovery for Configuration $2$ as a function of $(\theta_q, \theta_{\sigma})$. \\

\begin{figure}[!h]
\centering
\includegraphics[width=13.0cm]{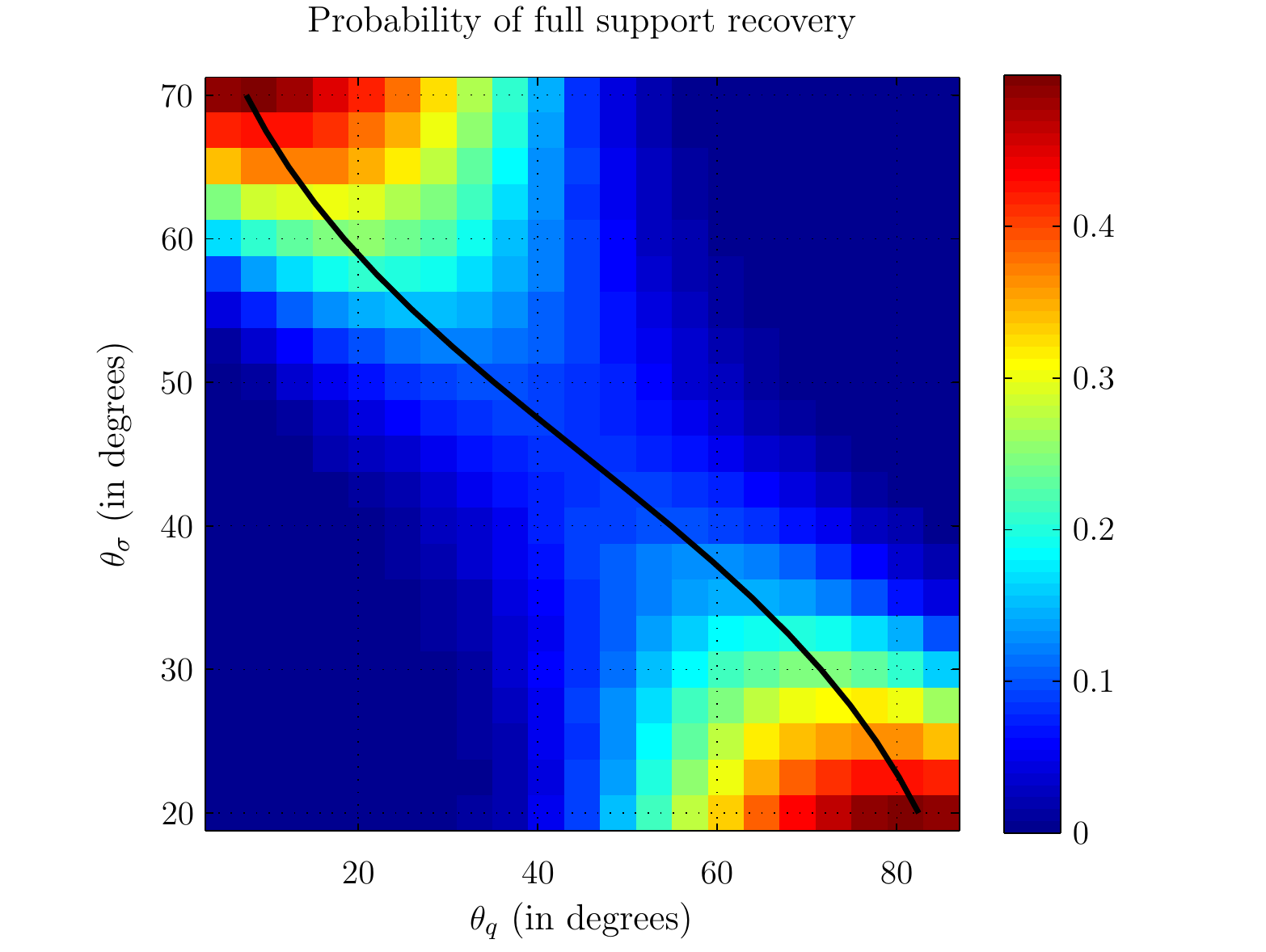}
\caption{Simulation results ($K = 2$) for Configuration $2$ (sign pattern $1$, $|S| = 30$, $\mu_X = 3.19$ and $\mathrm{SNR}_{\mathrm{in}} = 5.37$ dB) -- Probability of full support recovery as a function of $\theta_q$ and $\theta_{\sigma}$ -- The black curve refers to the optimal weights derived from Equation~(\ref{eq:B2}) -- Each pixel of the figure has been computed on the basis of $1.0 \; 10^4$ simulation cases}%
\label{fig:simAnglesPlotK2Conf2}%
\end{figure}

The first objective we would like to fulfill is demonstrating that SOMP-NS is capable to outperform SOMP whenever the noise standard deviations are not identical for each measurement vector. To do so, we will examinate, for each configuration and for each value of $\theta_{\sigma}$, what is the ratio of the probability of failure obtained for $\theta_q = 45^{\circ}$, \textit{i.e.}, the weights corresponding to SOMP, to the lowest probability of failure, \textit{i.e.}, that obtained for the truly optimal weights. Figure~\ref{fig:simAnglesSOMPvsOptimal} plots the aforementioned quantity for the all the configurations of Table \ref{tab:configsSimK2}. Note that the optimal weights are determined on the basis of the numerical results, the formula $q_k = 1/\sigma_k^2$ obtained on the basis of Equation~(\ref{eq:B2}) is not used. \\
\begin{figure}[!h]%
    \centering
    \hspace*{-23mm}\subfloat[Results for sign pattern $1$]{{\includegraphics[width=11.0cm]{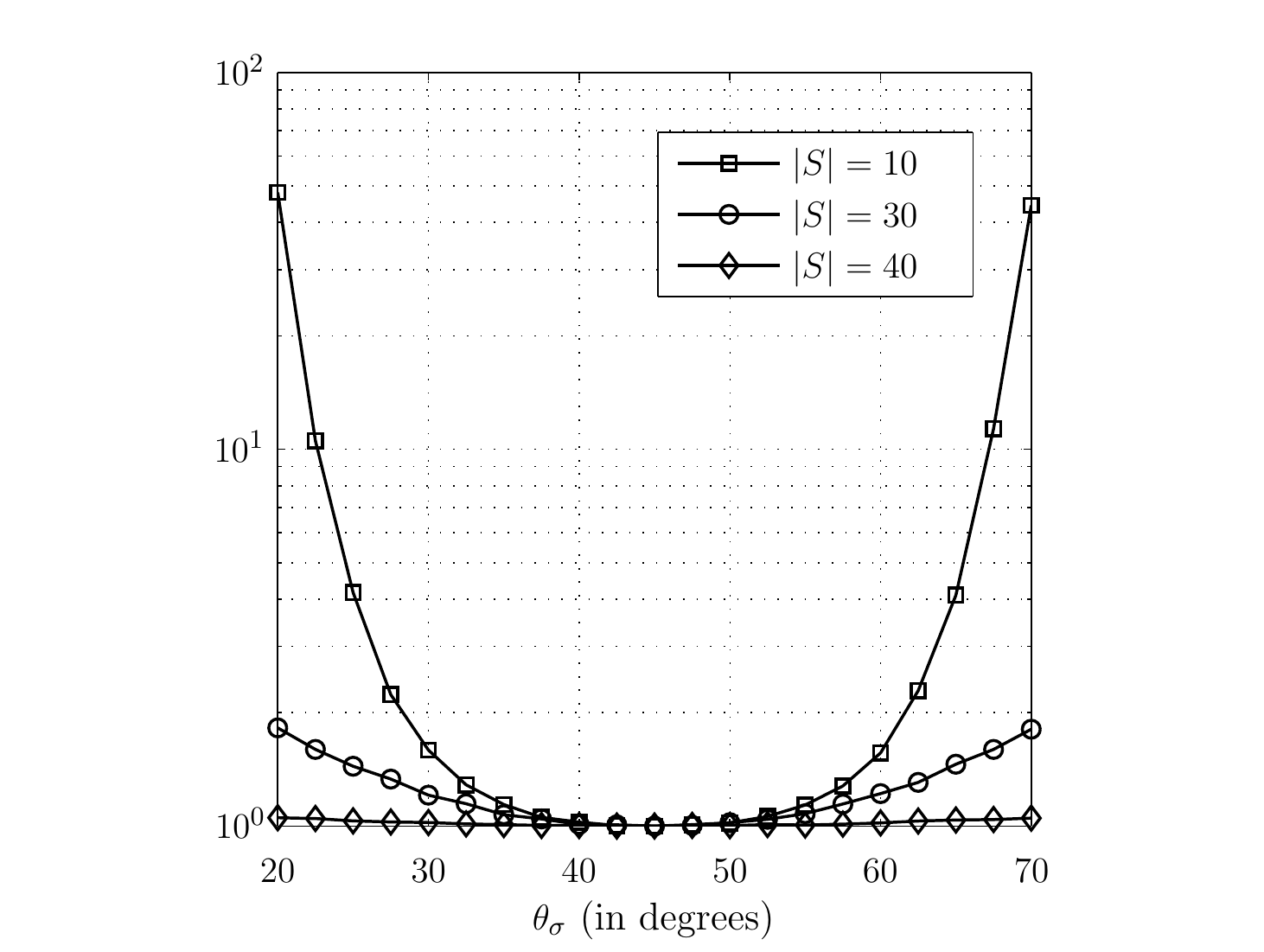} }}
    \hspace*{-20mm}\subfloat[Results for sign pattern $2$]{{\includegraphics[width=11.0cm]{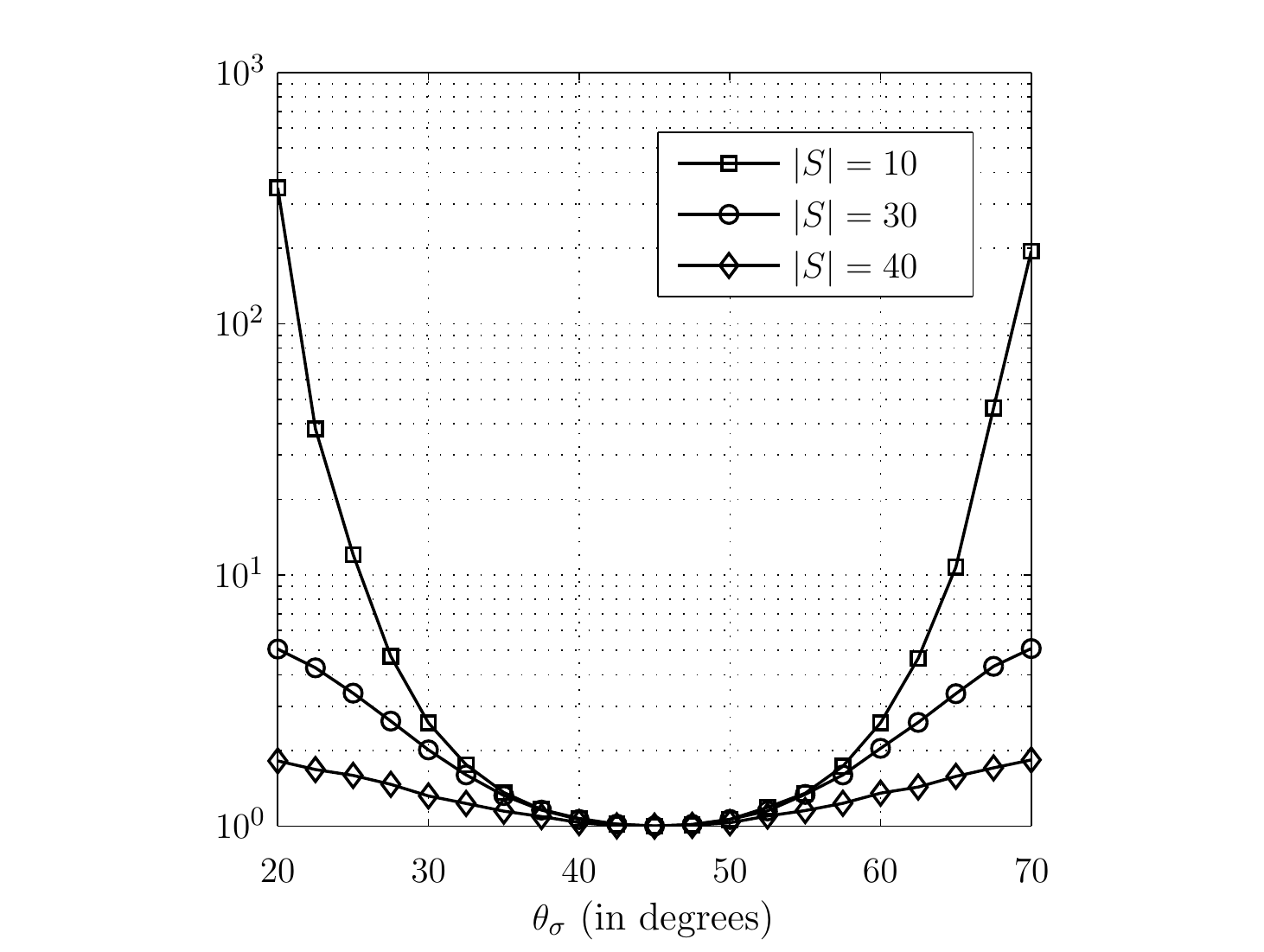} }}%
    \caption{Simulation results ($K = 2$) -- Ratio of the probability of failure of SOMP to that obtained for SOMP-NS when using the optimal weights}%
    \label{fig:simAnglesSOMPvsOptimal}%
\end{figure}

It is observed that the gains provided by the proper application of SOMP-NS are significant, especially for low values of the support size. The only case for which the gain is almost nonexistent is for sign pattern $1$, $\mu_X = 2.50$ and $|S| = 40$, \textit{i.e.}, Configuration $3$. This may actually be a consequence of the normalization procedure of $\mu_X$ described previously which ensures that the probability of successful support recovery for $(\theta_q, \theta_{\sigma}) = (45^{\circ}, 45^{\circ})$ is equal to $0.1$. The value of $\mu_X$ had to be chosen significantly higher than that of the other cases to attain the $0.1$ goal, which limits the impact of the noise and thus hinders the improvement of the performance by modifying the weights.\\

Let us now attack the second objective of Section~\ref{sec:numresults}. We wish to show that the formula $q_k = 1/\sigma_k^2$ obtained in Section~\ref{subsec:specialSigModel} delivers reliable estimates of the optimal weights. Figure~\ref{fig:simAnglesOptimalvsTheory} summarizes the results. First of all, it is observed that, for the first sign pattern, the solution $q_k = 1/\sigma_k^2$ always corresponds to the truly optimal weights while this is not true for the second sign pattern. In particular, the discrepancy between the numerical results and the theoretical formula (\ref{eq:B2}) increases as the size of the support augments.\\

\begin{figure}[!h]%
    \centering
    \hspace*{-23mm}\subfloat[Results for sign pattern $1$]{{\includegraphics[width=11.0cm]{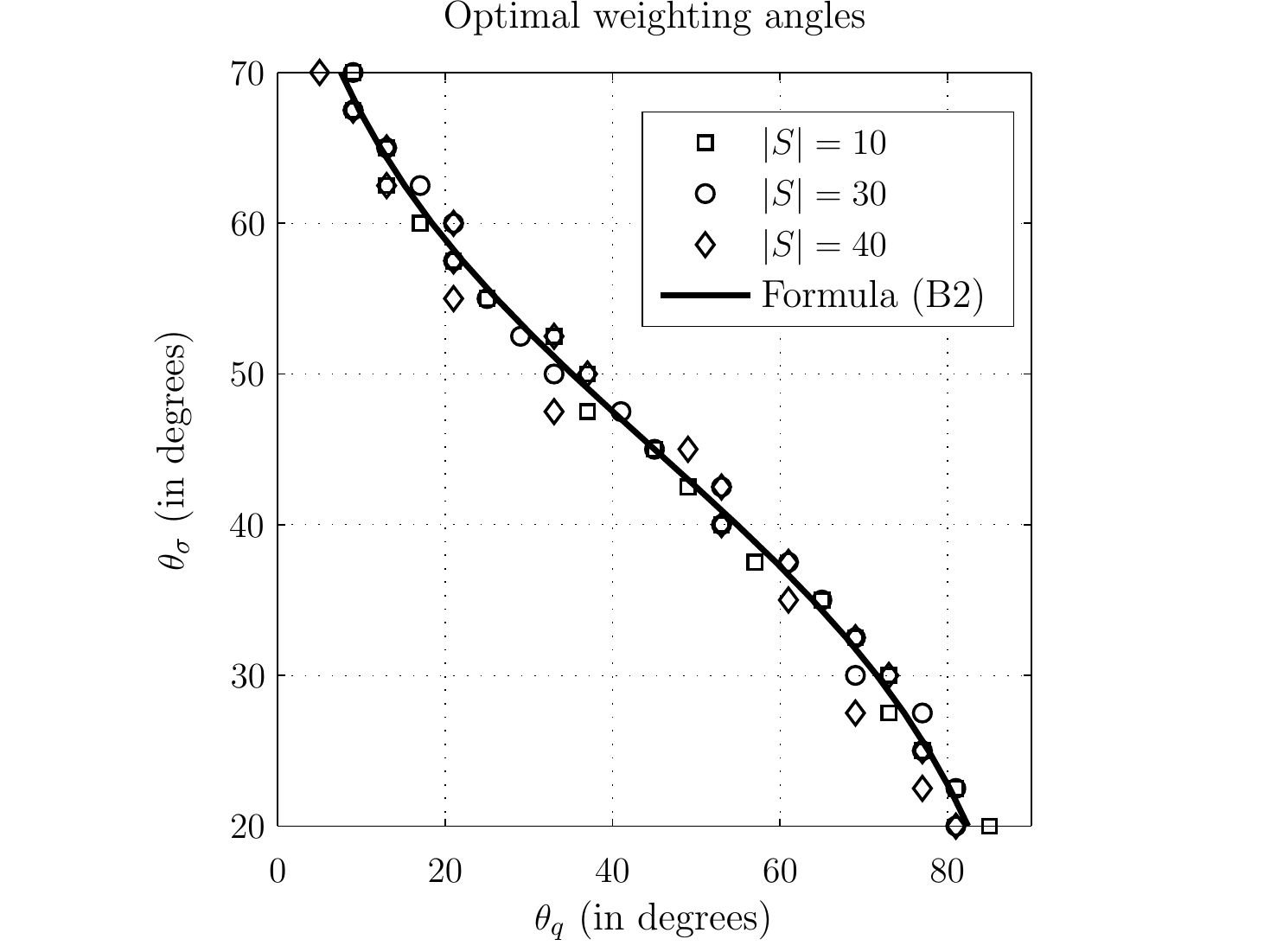} }}
    \hspace*{-20mm}\subfloat[Results for sign pattern $2$]{{\includegraphics[width=11.0cm]{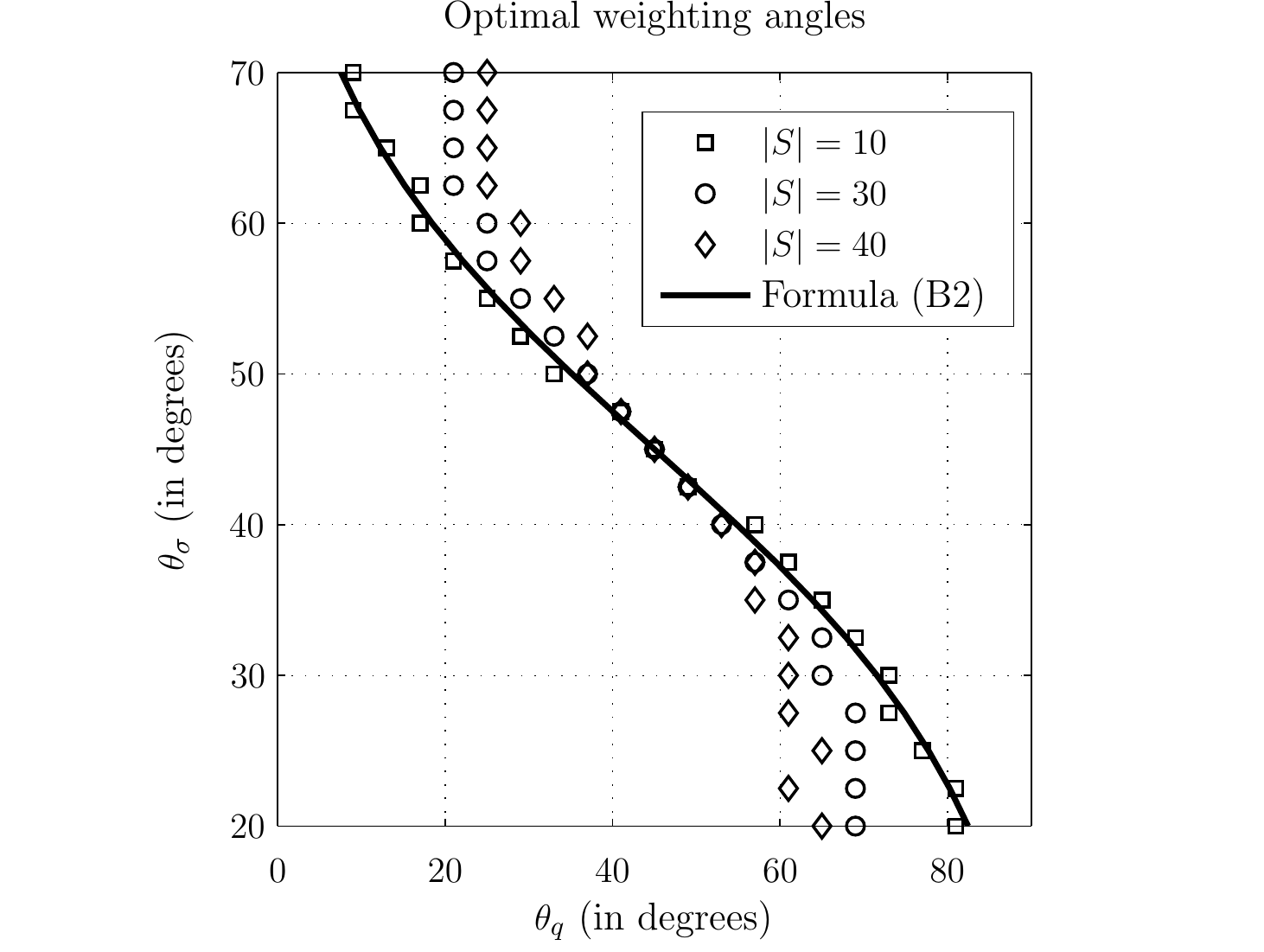} \label{subfig:SP2OptAngles} }}%
    \caption{Simulation results ($K = 2$) -- Optimal weighting angles -- The black curve denotes the optimal weights predicted by Equation~(\ref{eq:B2}) -- The markers correspond to the optimal weights derived from the simulation results}%
    \label{fig:simAnglesOptimalvsTheory}%
\end{figure}

The first observation is explained by the different sign patterns for each measurement vector. Although the weights have been introduced to better filter the influence of the noise, they also have an impact on the relative importance of each $\bsy{x}_k$ in the decisions that are taken. Given that the sparse vectors $\bsy{x}_k$ to be recovered  have identical distributions, it is to be expected that, without noise, the optimal weights are obtained by choosing $\theta_q = 45^{\circ}$ for symmetry reasons. Figure~\ref{fig:simAnglesPlotK2Conf6NoNoise} displays the simulation results obtained for a configuration identical to Configuration $6$ except that $\mathrm{SNR}_{\mathrm{in}} = 78.73$ dB, \textit{i.e.}, the influence of the noise is negligible. It is observed that the interaction of the weights and the sparse vectors to be recovered exists and that the optimal weighting angle is equal to $45^{\circ}$. A possible interpretation of the observations above is that the optimal weighting strategy is a mixture of the strategy that optimizes the support recovery in the noiseless case and of that which minimizes the impact of the noise on the decisions that are taken, \textit{i.e.}, $q_k = 1/\sigma_k^2$. Nevertheless, further theoretical developments should be conducted to assess whether the proposed interpretation is correct. Finally, it is worth pointing out that the phenomenon described above is not observed for sign pattern $1$ because $\bsy{x}_1$ and $\bsy{x}_2$ are identical in this case.\\

The reasons that explain why the optimal weights get closer to $\theta_q = 45^{\circ}$ when the support size increases, as shown in Figure~\ref{subfig:SP2OptAngles}, is not clear and would require additional theoretical investigations that fall outside of the framework of this paper.\\

\begin{figure}[!h]
\centering
\hspace*{-8mm}
\includegraphics[width=13.0cm]{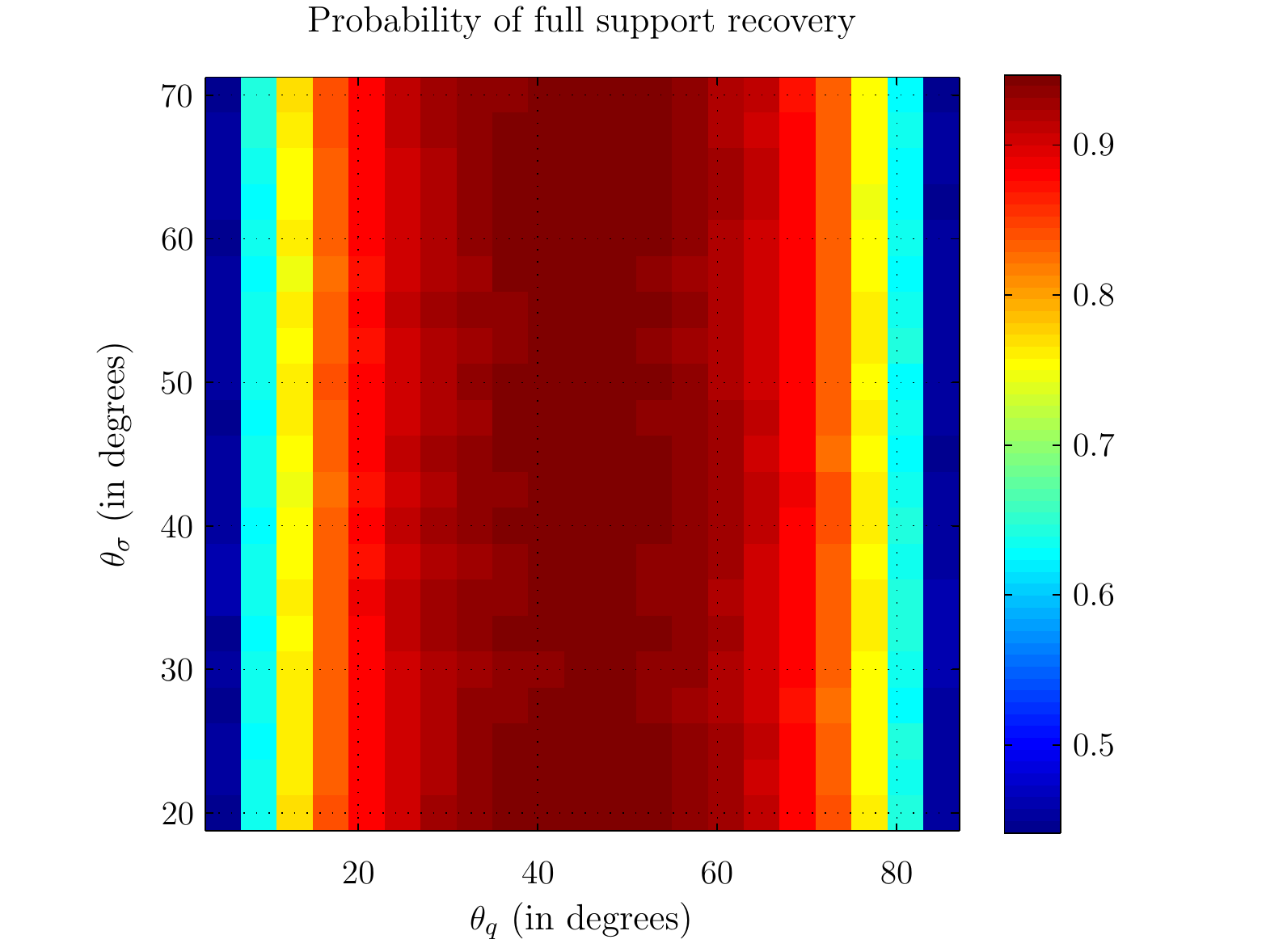}
\caption{Simulation results ($K = 2$) for Configuration $6$ with \mbox{$\mathrm{SNR}_{\mathrm{in}} = 78.73$ dB} -- Probability of full support recovery as a function of $\theta_q$ and $\theta_{\sigma}$ -- Sign pattern $2$ -- $|S| = 40$ -- $\mu_X = 15280$ -- \# cases $= 1.6 \; 10^4$}%
\label{fig:simAnglesPlotK2Conf6NoNoise}%
\end{figure}

\subsection{Simulation results for increasing values of $K$}

The final question of whether the proposed theoretical analysis properly conveys the properties of SOMP-NS whenever $K$ increases is to be discussed in this section. Our main objective is to show that, as predicted by Equation~(\ref{eq:B2}), the probability of failure of SOMP-NS decreases linearly with $K$ when it is plotted is semi-logarithmic axes. Indeed, Equation~(\ref{eq:B2}) yields
\begin{equation*}
\log \mathbb{P}_{\mathrm{fail}} \leq  \log \left( n \mathcal{C}_{|S|-1} \right) - K  \dfrac{\langle \bsy{q} \rangle^2 \mu_X^2}{2 \langle (q_k^2 \sigma_k^2)_{k \in \lbrack K \rbrack} \rangle} \varepsilon'^2
\end{equation*}
where $\mathbb{P}_{\mathrm{fail}}$ denotes the probability of failure of correct support recovery.\\

The simulation framework consists of a fixed weighting strategy for which all the weights are equal, \textit{i.e.}, the weighting strategy corresponds to SOMP. For each value of $K$, the noise vector is given by $\bsy{\sigma} = (\sqrt{2}/2) \; (1, 1, \dots , 1)^{\mathrm{T}}$ so that it is reminiscent of the noise vector defined in Section~\ref{subsec:simresK2} for $K = 2$. The results are plotted in Figure~\ref{fig:simPlotK}. The configurations that have been chosen are directly inspired of those presented in Table~\ref{tab:configsSimK2}. The number of cases simulated for each curve and each value of $K$ is equal to $4\; 10^5$. Some configurations have been discarded because they do not exhibit a probability of failure equal to $0$ without noise, which is incompatible with the implicit assumption of our theoretical model that no errors are committed in the noiseless case. Indeed, if $1 - \|\bsy{\Phi}_S^+ \bsy{\Phi}_{\overline{S}} \|_{1} \leq 0$, \textit{i.e.}, the ERC is the noiseless case is not satisfied, then Equation~(\ref{eq:SOMPanalysis2}) cannot hold.  For example, Figure~\ref{fig:simAnglesPlotK2Conf6NoNoise} shows that Configuration $6$ from Table \ref{tab:configsSimK2} exhibits a non-zero probability of failure without noise.\\

\begin{figure}[!h]
\centering
\hspace*{-13mm}
\includegraphics[width=13.0cm]{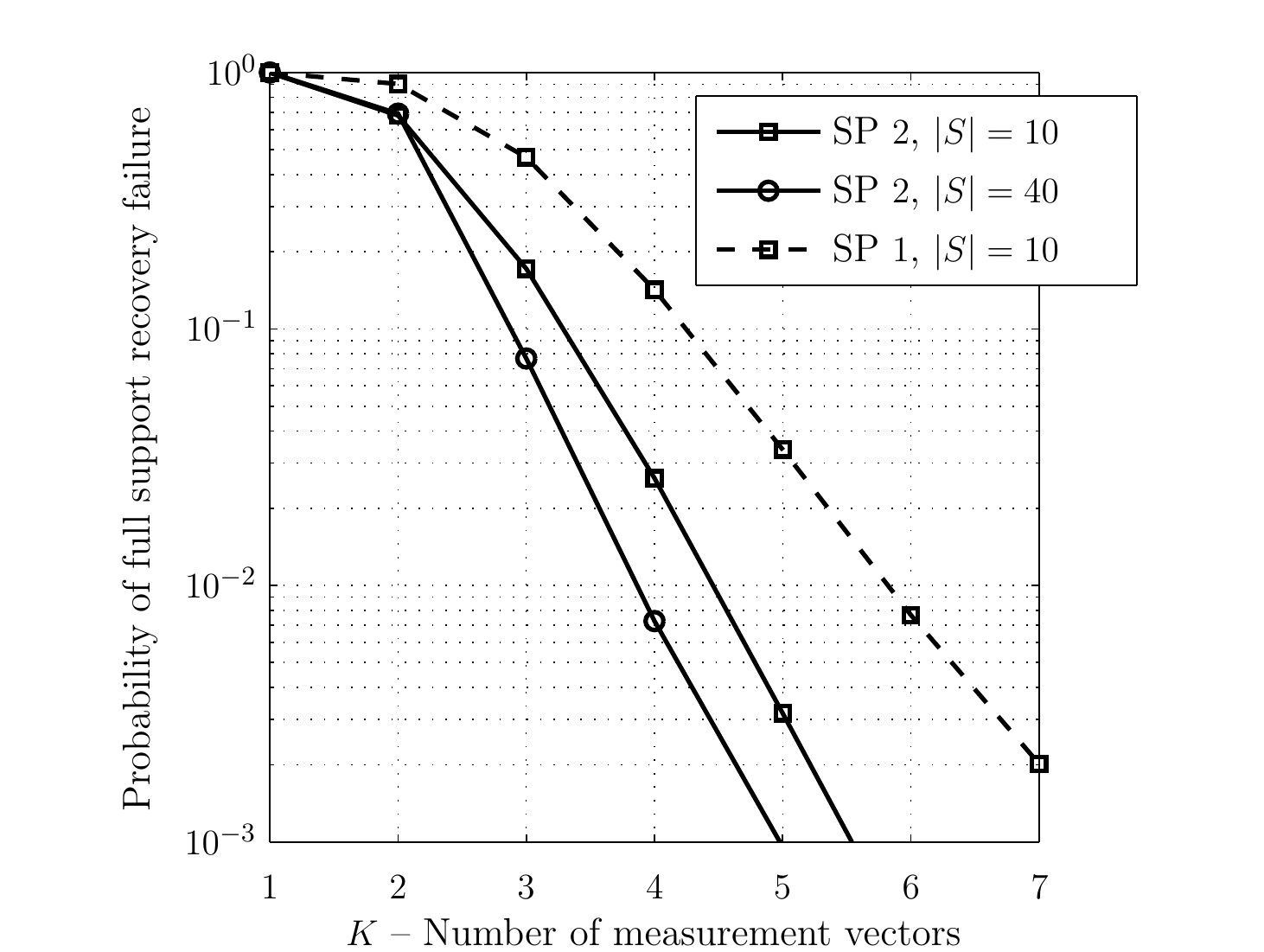}
\caption{Probability of full support recovery failure as a function of $K$ -- SP refers to sign pattern}%
\label{fig:simPlotK}%
\end{figure}

The principal observation for Figure~\ref{fig:simPlotK} is that the slope of $\mathbb{P}_{\mathrm{fail}}$ in semi-logarithmic axes is linear with regards to $K$. This observation provides evidence that the theoretical model conveys the behavior of SOMP-NS when $K$ increases.

\subsection{Summary of the numerical results}\label{subsec:summaryNumres}

The numerical results have revealed the following interesting facts:
\begin{enumerate}
\item SOMP-NS provides significant performance improvements when compared to SOMP provided that the weights are properly chosen and that the noise variances are different for each measurement vector.
\item The formula $q_k = 1/\sigma_k^2$ derived from Equation~(\ref{eq:B2}) corresponds to the truly optimal weights whenever
\begin{itemize}
\item The sparse vectors $\bsy{x}_k$ to be recovered are identical.
\item The support size $|S|$ is low enough.
\end{itemize}
The exact reason why the formula $q_k = 1/\sigma_k^2$ gets less accurate as the size of the support increases remains an open question.
\item The theoretical analysis properly predicts the characteristics of the decrease of the probability of failure of SOMP-NS whenever the number of measurement vectors increases.
\end{enumerate}

Besides the three points above, which answer the three questions enumerated at the beginning of Section~\ref{sec:numresults}, the close fitting of $q_k = 1/\sigma_k^2$ and the numerical results for sign pattern $1$ suggests that the bias term $b(\bsy{q}, \bsy{\sigma})$ is indeed an artifact of our developments as adding it would change the theoretically optimal weights and we would then observe a mismatch between them and those obtained by simulation.

\section{Future work}
As suggested in Section~\ref{subsec:conjectures}, the bias $b(\bsy{q}, \bsy{\sigma})$ could be removed by avoiding to make use of the inequalities (\ref{eq:firstBigApprox}) and (\ref{eq:secondBigApprox}). Using a more subtle approach than the use of the union bound could also yield performance improvements.\\

Although it appears to be difficult, replacing the term $\mathcal{C}_{|S|-1}$ by a function that depends linearly on $|S|$ would be of great interest, especially since it would close a hole in the literature regarding the performance of the well-known SOMP algorithm that is a particular case of ours.\\

Finally, extending the presented analysis by performing a joint statistical analysis of both the noise and the sparse signals to be recovered would be of great interest. To begin with, it would provide a theoretical model that predicts the truly optimal weights by simultaneously taking into account how they impact the sparse signals to be recovered and the noise vectors. Next, it would also enable one to comprehend why the discrepancy between the formula $q_k = 1/\sigma_k^2$ and the truly optimal weights increases as the support size augments (see Figure~\ref{fig:simAnglesOptimalvsTheory}). Finally, the statistical analysis of the sparse signals could replace  $n \mathcal{C}_{|S|-1}$ by $\overline{n} \mathcal{C}_{|S|-1}$ (where $\overline{n}$ is significantly lower than $n$) as conjectured in Section~\ref{subsec:conjectures}.

\section{Conclusion}
A novel algorithm entitled SOMP-NS that generalizes SOMP by associating weights with each measurement vector has been proposed. A theoretical framework to analyze this algorithm has been built. Lower bounds on the probability of full support recovery by means of SOMP-NS have been developed in the case where the noise corrupting the measurements is Gaussian. Numerical simulations have revealed that the developed theoretical results accurately depict key components of the behavior of SOMP-NS while they also fail to capture some of its properties. In particular, it has been shown that, under the right circumstances, the weights of SOMP-NS can be efficiently optimized on the basis of the proposed theoretical bounds. Finally, the reasons that explain why some characteristics of SOMP-NS are not properly conveyed by the theoretical analysis have been discussed and potential workarounds to be investigated have been suggested.

\end{document}